\documentclass[11pt]{article}

\usepackage[parfill, skip=.25\baselineskip plus 1pt, indent=17pt]{parskip}

\usepackage[a4paper, margin=1in]{geometry}

\usepackage{amsmath,array}
\usepackage{amssymb}
\usepackage{amsthm}
\usepackage[inline,shortlabels]{enumitem}
\usepackage{graphicx}
\usepackage{hyperref}
\usepackage{times}
\usepackage[cal=boondoxo]{mathalfa}
\usepackage{upgreek}

\allowdisplaybreaks

\setlist[enumerate,1]{label=\roman*., leftmargin=8mm}
\setlist[itemize,1]{label=\textbullet, leftmargin=6mm}

\newtheorem{theorem}{Theorem}[section]
\newtheorem{corollary}[theorem]{Corollary}
\newtheorem{remark}{Remark}[theorem]
\newtheorem{proposition}[theorem]{Proposition}
\newtheorem{lemma}[theorem]{Lemma}
\newtheorem{definition}[theorem]{Definition}

\numberwithin{equation}{section}

\newcommand{\<}{\langle}
\renewcommand{\>}{\rangle}

\DeclareMathOperator{\Span}{span}

\newcommand{\norm}[1]{\left\lVert #1\right\rVert}
\newcommand{\abs}[1]{\left\lvert #1\right\rvert}
\newcommand{\brac}[1]{\left( #1\right)}
\newcommand{\sbrac}[1]{\left[ #1\right]}

\newcommand{\set}[1]{\left\{#1\right\}}
\newcommand{\bs}[1]{\boldsymbol{\mathbf{#1}}}
\newcommand{\mat}[1]{\begin{pmatrix}#1\end{pmatrix}}

\newcommand{\ph}{\,\cdot\,}

\newcommand{\Z} {\mathbb{Z}}

\newcommand{\C} {\mathbb{C}}
\newcommand{\R} {\mathbb{R}}

\newcommand{\lpc}{\Delta}
\newcommand{\grad}{\nabla}

\newcommand{\mb}{\mathbf}

\renewcommand{\d}{\textup{d}}

\renewcommand{\Re}{\operatorname{Re}}
\renewcommand{\Im}{\operatorname{Im}}

\renewcommand{\L}{\mathcal L}

\newcommand{\nmodes}{n^u}
\newcommand{\Mor}{n^-}

\title{A turning point principle for liquid Lane--Emden stars}

\author{King Ming Lam\thanks{Delft Institute of Applied Mathematics, Delft University of Technology, 2628 CD Delft, Netherlands. Email: \href{mailto:K.M.Lam@tudelft.nl}{K.M.Lam@tudelft.nl}.}}

\date{}

\begin{document}

\maketitle

\begin{abstract}
Upon fixing the adiabatic index, the spherically symmetric liquid Lane--Emden stars governed by the Euler--Poisson system with a ``stiffened gas'' equation of state $p=\rho^\gamma-1$ form a one-parameter family, naturally parametrised by the central density $\kappa=\bar\rho_\kappa(0)\in(1,\infty)$. In contrast to the gaseous case, the liquid free boundary breaks the self-similarity of the family and bends the mass--radius curve, creating extrema of the mass when $\gamma<2(d-1)/d$. We formulate a turning point principle in the spirit of Zel'dovich, Wheeler and Thorne, and of its rigorous relativistic counterpart from the works of Had\v{z}i\'c, Lin and Rein: the number of radial growing modes equals the negative Morse index of the linearised operator, is locally constant along the family, and can change only at the turning points of the mass--radius curve, the bending orientation there dictating whether a growing mode is gained or lost. We describe the large central density limit in which the gaseous tail --- read off from a planar dynamical system --- determines whether the mass--radius curve spirals, and hence whether the number of growing modes remains bounded or tends to infinity. As corollaries we recover and sharpen the known radial (in)stability results for liquid Lane--Emden stars and, in combination with existing non-radial analysis, obtain a turning point criterion for non-radial stability.
\end{abstract}

\tableofcontents

\section{Introduction}

\subsection{The Euler--Poisson system and liquid Lane--Emden stars}

We consider the free-boundary Euler--Poisson system, a fundamental model of a self-gravitating compressible fluid. The unknowns are the fluid density $\rho\ge0$, the velocity $\mb u$, the pressure $p\ge0$, and the gravitational potential $\phi$, solving
\begin{alignat}{3}
\partial_t\rho+\grad\cdot(\rho\mb u) &=0
\qquad &&\text{ in } \ \Omega(t), \label{E:CONT}\\
\rho(\partial_t + \mb u\cdot\nabla)\mb u + \grad p + \rho\grad\phi &=\mb 0 
\qquad &&\text{ in } \ \Omega(t), \label{E:MOM}\\
\lpc\phi &=4\pi\rho 
\qquad &&\text{ in } \ \R^d,\label{E:GRAV}
\end{alignat}
with the asymptotic condition $\lim_{|\mb x|\to\infty}\phi(t,\mb x)=0$ that models an isolated star. Here $d\ge3$, $\Omega(t):=\{\mb x:\rho(t,\mb x)>0\}$ is the (moving) interior of the support of the star, and we close the system with the liquid polytropic equation of state
\begin{equation}\label{E:EOS}
  p=\rho^\gamma-1,\qquad\text{where}\qquad \gamma\in[1,2).
\end{equation}
At the vacuum free boundary $\partial\Omega(t)$ we impose that the pressure matches the vacuum, $p=0$ on $\partial\Omega$, and that the boundary moves with normal velocity $\mb u\cdot\mb n$. We refer to~\eqref{E:CONT}--\eqref{E:EOS} as the liquid (EP)$_\gamma$-system. The constant in~\eqref{E:EOS} has been normalised to $1$, so that the density jumps from $0$ in the vacuum to the value $1$ on the liquid surface, where $p=0$.

\begin{remark}[Normalisation of the gravitational coupling]\label{R:coupling}
In~\eqref{E:GRAV} we keep the three-dimensional constant $4\pi$ in every dimension $d\ge3$, following~\cite{Lam_2024}, rather than the geometrically natural constant $\upomega_d:=|\mathbb S^{d-1}|=2\pi^{d/2}/\Gamma(d/2)$, i.e. the surface area of the unit sphere in $d$ dimensions. This is just a choice of the normalisation of the gravitational constant (physically, just a choice of units), our results in this paper are not affected by this choice, as different choices are related by rescaling. Indeed, every statement we prove is invariant under such choices --- the growing-mode counts, the $\kappa$-locations and orientations of the turning points, the thresholds $\gamma_*,\gamma_\sharp$, and the discriminant~\eqref{E:DISC}, which accordingly contains no $\pi$ --- only the explicit constants in displayed formulas change.
\end{remark}

The simplest non-trivial solutions are the time-independent, spherically symmetric steady states, the \emph{liquid Lane--Emden stars}, in which pressure and self-gravity are in exact balance. They represent stars in hydrostatic
equilibrium. Writing $r=|\mb x|$, the steady momentum equation $\grad p+\rho\grad\phi=\mb 0$ reduces to the ODE
\begin{align}\label{E:LEODE}
0=
\begin{cases}
\dfrac{\gamma}{\gamma-1}\lpc\bar\rho^{\gamma-1}+4\pi\bar\rho &\quad\text{ when }\quad\gamma>1,\\[2mm]
\lpc\ln\bar\rho+4\pi\bar\rho &\quad\text{ when }\quad\gamma=1,
\end{cases}
\qquad \grad\bar\rho(\mb 0)=\mb 0,
\end{align}
where $\lpc=r^{1-d}\partial_r(r^{d-1}\partial_r)$. For each central value $\bar\rho(0)>0$ this ODE has a unique decreasing solution $f_0$, and a \emph{liquid} star is obtained by truncating $f_0$ at the radius where it reaches the surface density $1$: setting $\bar\rho:=f_0\mb 1_{[0,R]}$ with $f_0(R)=1$. Only profiles with $\bar\rho(0)>1$ produce a genuine liquid star, so the family is naturally parametrised by the central density
\begin{align}\label{E:KAPPA}
\kappa:=\bar\rho_\kappa(0)\in(1,\infty),
\end{align}
which plays here the role that the central redshift plays for relativistic stars. As $\kappa\searrow1$ the star shrinks to a point (small relative central density), while $\kappa\to\infty$ is the highly compressed regime. Details of the liquid Lane--Emden stars can be found in~\cite{Lam_2024}.

\subsection{Turning point principles and related literature}

In our universe, we observe both stable stars like our
Sun and unstable stars such as exploding supernovae and stars collapsing into black holes. Thus an important question in astrophysics is that of stellar stability.

A recurring theme in the stability theory of self-gravitating bodies is that dynamical (in)stability can be diagnosed from bulk, observable quantities rather than from a direct spectral computation. For relativistic stars this is the content of the \emph{turning point principle}, also called the static criterion or the $M(R)$-method. Fixing an equation of state, the spherically symmetric steady states of the Einstein--Euler system form a one-parameter family parametrised by the central redshift $\kappa$, and one plots the mass $M_\kappa$ against the radius $R_\kappa$. Zel'dovich~\cite{Zeldovich_1963} and Wheeler (see Harrison, Thorne, Wakano and Wheeler~\cite{Harrison_Thorne_Wakano_Wheeler_1965}, pp.~60--66) proposed that along such a family stability can be exchanged for instability, and vice versa, only at the extrema of $\kappa\mapsto M_\kappa$, and that the orientation of the bend at such an extremum determines whether a growing mode is gained or lost. The proposal grew out of the first systematic linear stability theory for relativistic stars: Chandrasekhar~\cite{Chandrasekhar_1964} formulated the radial spectral problem variationally, and~\cite{Harrison_Thorne_Wakano_Wheeler_1965} (see also Bardeen, Thorne and Meltzer~\cite{Bardeen_Thorne_Meltzer_1966}) recognised the equivalent characterisation of spectral stability through the second variation of the mass at fixed baryon number --- the variational germ of every turning point statement. Refinements and heuristic derivations of the static approach go back to Thorne~\cite{Thorne_1966}, its textbook accounts include Zel'dovich and Novikov~\cite{Zeldovich_Novikov_1971} and Shapiro and Teukolsky~\cite{Shapiro_Teukolsky_1983}, and to this day the resulting criterion $\partial M/\partial\rho_c>0$ is the standard first stability test applied to neutron star models. An abstract account of turning point instabilities was given by Sorkin~\cite{Sorkin_1981, Sorkin_1982}; it underlies the turning point criterion of Friedman, Ipser and Sorkin~\cite{Friedman_Ipser_Sorkin_1988} for the axisymmetric secular stability of uniformly rotating relativistic stars, and the relation between such thermodynamic turning points and the dynamical ones considered here is discussed by Schiffrin and Wald~\cite{Schiffrin_Wald_2014}.

For half a century the principle remained a remarkably successful heuristic; it was placed on a rigorous footing only recently. Lin and Zeng~\cite{Lin_Zeng_2022} introduced a separable Hamiltonian framework for the linearised Euler--Poisson system (resting on their general index theory for linear Hamiltonian PDEs~\cite{Lin_Zeng_Memoirs_2022}) and proved the turning point principle for Newtonian gaseous stars along one-parameter families of steady states with general equations of state. Building on this framework, Had\v{z}i\'c, Lin and Rein~\cite{Hadzic_Lin_Rein_2021} proved for the Einstein--Euler system that steady states of small central redshift are spectrally stable while strongly relativistic ones --- those with very large central redshift --- are unstable, together with an exponential trichotomy for the linearised flow, and Had\v{z}i\'c and Lin~\cite{Hadzic_Lin_2021} then established the full turning point principle. In the relativistic setting the number of growing modes is expressed as the difference of the negative Morse index of a Schr\"odinger-type reduced operator and a winding index that reads off the bending of the mass--radius curve; one consequence is that the number of growing modes tends to infinity as the central redshift grows and the mass--radius curve spirals. The spiral itself has an independent history: for the Tolman--Oppenheimer--Volkoff equation it was analysed by Makino~\cite{Makino_2000} and, by numerical and dynamical-systems methods, by Nilsson and Uggla~\cite{Nilsson_Uggla_2000} and Heinzle, R\"ohr and Uggla~\cite{Heinzle_Rohr_Uggla_2003}; for the parallel dynamical-systems theory of Newtonian stellar models see Heinzle and Uggla~\cite{Heinzle_Uggla_2003} and also~\cite{Heinzle_2002}. The planar system that we extract in Section~\ref{S:proof} to govern the gaseous tail of the liquid family belongs to this circle of ideas; what the turning point principle adds is that this system encodes not merely the shape of the mass--radius curve but the exact growing-mode count along it.

It should be stressed that a turning point principle is a genuinely dynamical assertion, not a soft consequence of the geometry of the family. For the gravitational Vlasov--Poisson system, the kinetic sibling of the Euler--Poisson system, Ramming and Rein~\cite{Ramming_Rein_2017} proved that the mass--radius relation of the King and related steady state families is a spiral; yet these steady galaxies are nonlinearly stable~\cite{Guo_Rein_2001, Lemou_Mehats_Raphael_2012}, so the Newtonian kinetic mass--radius spiral carries no instability at all and the turning point principle in the above form \emph{fails}. In the relativistic kinetic setting, Zel'dovich and Podurets~\cite{Zeldovich_Podurets_1965} proposed a binding-energy turning point criterion for the Einstein--Vlasov system; the instability of strongly relativistic galaxies was proven in~\cite{Hadzic_Lin_Rein_2021}, but the recent numerics of G\"unther, Straub and Rein~\cite{Gunther_Straub_Rein_2021} give strong evidence against the classical binding-energy hypothesis, exhibiting stable configurations beyond the first binding-energy maximum; see the review~\cite{Rein_2023}. Whether a given one-parameter family of self-gravitating equilibria obeys an exact turning point law is therefore a delicate question that must be settled model by model, and settling it completely for a natural Newtonian free-boundary family is one purpose of the present paper.

In this paper we formulate and prove the corresponding turning point principle for the Newtonian \emph{liquid} Lane--Emden stars. The analogy with the relativistic setting is more than formal. For the classical gaseous Lane--Emden stars~\cite{Lane_1870, Emden_1907, Chandrasekhar_1939} the family is exactly self-similar, the mass--radius curve is a single power law without extrema, and the stability is therefore independent of the central density: the stars are linearly stable precisely when $\gamma\ge 2(d-1)/d$, which in dimension $d=3$ is the classical threshold $\gamma\ge4/3$ of Lin~\cite{Lin_1997} and Jang and Makino~\cite{Jang_Makino_2020}. The turning point principle of~\cite{Lin_Zeng_2022} is thus consistent with, but degenerate along, this family; genuine Newtonian turning points require the scale invariance to be broken, as it is for the white dwarf equation of state with its limiting Chandrasekhar mass~\cite{Chandrasekhar_1931}. The nonlinear theory reflects the same threshold: nonlinear instability for $6/5\le\gamma<4/3$ was proven by Jang~\cite{Jang_2008, Jang_2014}, with the expanding dynamics near the unstable stars recently detailed by Cheng, Cheng and Lin~\cite{Cheng_Cheng_Lin_2025}, while nonlinear stability for $4/3<\gamma<2$, conditional on global existence, was obtained variationally by Rein~\cite{Rein_2003} and Luo and Smoller~\cite{Luo_Smoller_2009} and made unconditional for spherically symmetric perturbations by Lin, Wang and Zhu~\cite{Lin_Wang_Zhu_2024}. At the critical exponent $\gamma=4/3$ the Lane--Emden star is embedded in the family of collapsing and expanding Goldreich--Weber stars~\cite{Goldreich_Weber_1980, Makino_1992} and is nonlinearly unstable despite its neutral linear stability, the expanding Goldreich--Weber stars being themselves nonlinearly stable~\cite{Hadzic_Jang_2018, Hadzic_Jang_Lam}.

Imposing the liquid surface density truncates the self-similar profile and bends the mass--radius curve, creating genuine mass extrema once $\gamma<2(d-1)/d$. The density jump at the liquid surface places the problem outside the gaseous framework: the linearised operator acquires a surface term and a Robin boundary condition (see~\eqref{E:TE} below), and it is exactly this boundary contribution that produces the turning points. In~\cite{Lam_2024} the author showed that this renders the liquid stars linearly stable at small relative central density and, in dimensions $d<10$, linearly unstable at large central density when $\gamma<2(d-1)/d$, in striking parallel with the relativistic dichotomy of~\cite{Hadzic_Lin_Rein_2021}; nonlinear instability at large central density was subsequently obtained by Hao and Miao~\cite{Hao_Miao_2024} taking the linear analysis as input, and the non-radial linear theory was developed in~\cite{Lam_nonradial}. The liquid (``stiffened gas'') equation of state~\eqref{E:EOS} is a standard model of nearly incompressible fluids; at $\gamma=1$ it reduces to $p=\rho-1$, the equation of state of the \emph{hard phase} of Christodoulou's two-phase model of gravitational collapse~\cite{Christodoulou_1995}, itself modelled on Zel'dovich's stiff equation of state for matter at ultrahigh density~\cite{Zeldovich_1962}; for relativistic free-boundary problems of this liquid type see~\cite{Oliynyk, Miao_Shahshahani_Wu_2021}. The Cauchy problem underlying our linear theory is by now well understood: local well-posedness for the liquid Euler--Poisson system~\eqref{E:CONT}--\eqref{E:EOS} was proven by Ginsberg, Lindblad and Luo~\cite{Ginsberg_Lindblad_Luo_2020}, following the liquid free-boundary Euler theory of Lindblad~\cite{Lindblad_2005}, Trakhinin~\cite{Trakhinin_2009} and Coutand, Hole and Shkoller~\cite{Coutand_Hole_Shkoller_2013}; for the gaseous counterparts see~\cite{Coutand_Shkoller_2012, Jang_Masmoudi_2015, Gu_Lei_2016, Hadzic_Jang_2019, Ifrim_Tataru_2024}. A sharp linear theory of the kind organised by a turning point principle is precisely the input that such nonlinear frameworks require, as~\cite{Hao_Miao_2024} illustrates in the present context.

The turning point principle proven below organises these statements into a single counting law along the mass--radius curve, and exhibits the planar dynamical system governing the gaseous tail as the exact Newtonian analogue of the relativistic mass--radius spiral. Beyond recovering the (in)stability results of~\cite{Lam_2024}, it determines the precise number of growing modes at every non-critical central density, locates the onset of instability at the first mass maximum, and resolves the large central density asymptotics into a spiral/node dichotomy governed by the explicit discriminant~\eqref{E:DISC}, with dimension ten as the threshold. Two features seem worth highlighting. First, in our comoving formulation the conservation of mass is built into the admissible perturbations, so that the growing-mode count is the plain negative Morse index of the reduced operator: no mean-zero constraint is imposed and no winding index needs to be subtracted, in contrast with the constrained reductions of~\cite{Lin_Zeng_2022, Hadzic_Lin_Rein_2021, Hadzic_Lin_2021}, a simplification of the bookkeeping that may be useful in other free-boundary stability problems. Second, in combination with the non-radial analysis of~\cite{Lam_nonradial}, the principle upgrades to a criterion for \emph{non-radial} linear stability read off entirely from the mass--radius curve (Corollary~\ref{C:nonradial}); no analogous rigorous non-radial result is currently available for relativistic stars. In the 1920s, in his debate with Eddington, Jeans argued that liquid stars are more stable than gaseous ones~\cite{Jeans_1927, Jeans_1928, Eddington_1928}. His specific model was wrong, but the intuition has a modern counterpart: as found in~\cite{Lam_2024, Lam_nonradial}, liquid stars resemble relativistic stars more closely than gaseous ones do --- the scale-breaking constant in the pressure playing the role that the speed of light plays relativistically --- and the turning point principle established here, complete with its spiral, is one more manifestation of that resemblance.

\subsection{Informal statement of the main results}

Throughout we fix $d\ge3$ and $\gamma\in[1,2)$, and write
\begin{align}\label{E:GAMMASTAR}
\gamma_*:=\frac{2(d-1)}{d},\qquad \gamma_\sharp:=\frac{2d}{d+2},\qquad \gamma_\sharp<\gamma_*<2.
\end{align}
The threshold $\gamma_*$ is the mass critical index that separates stable from possibly unstable adiabatic indices; $\gamma_\sharp$ is the energy critical index that separates the gaseous profiles of compact support ($\gamma>\gamma_\sharp$) from those of infinite support ($\gamma\le\gamma_\sharp$). For each $\kappa\in(1,\infty)$ we denote by $(\bar\rho_\kappa,R_\kappa)$ the liquid Lane--Emden star of central density $\kappa$, by $M_\kappa$ its mass, and by $\nmodes(\kappa)$ the number of its radial growing modes; precise definitions are given in Section~\ref{S:results}.

Our first result, Theorem~\ref{T:COUNT}, identifies the number of growing modes with the negative Morse index of the linearised radial operator $\L_\kappa$ and pins down its kernel. In contrast with the relativistic case, we use a comoving (Lagrangian) formulation that builds in the conservation of total mass, so that no constraint needs to be imposed and no winding index is subtracted: the count is simply $\nmodes(\kappa)=\Mor(\L_\kappa)$. The kernel is non-trivial exactly at the turning points of $\kappa\mapsto M_\kappa$, the marginal mode being the variation of the steady family.

Our second result, Theorem~\ref{T:TPP}, is the turning point principle: $\nmodes(\kappa)$ is locally constant and changes only at the extrema of the mass, by $\pm1$ according to the bending orientation of the mass--radius curve, in exact analogy with~\cite{Hadzic_Lin_Rein_2021, Hadzic_Lin_2021}. It further describes the global shape of the curve: there are no extrema and permanent stability when $\gamma\ge\gamma_*$, and a first onset of instability at a mass maximum when $\gamma<\gamma_*$ and, in addition, $\gamma\ge\gamma_\sharp$ or the discriminant~\eqref{E:DISC} below is negative --- a hypothesis that is automatic in every dimension $d<10$. In the strictly infinite-support regime $\gamma<\gamma_\sharp$, the curve converges, as $\kappa\to\infty$, to the point representing the singular self-similar liquid star, and it does so along a spiral (with $\nmodes(\kappa)\to\infty$) precisely when the discriminant
\begin{align}\label{E:DISC}
\mathcal D(\gamma,d):=(d-2)^2(2-\gamma)^2-8(2(d-1)-d\gamma)
\end{align}
of the planar tail dynamics is negative, as it automatically is for every $d<10$; when $\gamma_\sharp\le\gamma<\gamma_*$ the curve instead returns to the origin with finitely many turning points and the growing-mode count stabilises, to exactly $1$ in the explicitly solvable borderline case $\gamma=\gamma_\sharp$. In the one remaining case --- the node regime $\gamma<\gamma_\sharp$, $\mathcal D(\gamma,d)\ge0$, possible only when $d\ge10$ --- the curve converges with an asymptotic tangent direction and the count stabilises to a finite value, but the stabilised value, and with it the large central density (in)stability, is left open; see Remark~\ref{R:open}.

We state these results in detail in Section~\ref{S:results}, and prove them in Section~\ref{S:proof} by a direct variational analysis of the reduced radial operator $\mathcal L_\kappa$, combined with an adaptation of the planar dynamical-systems method of~\cite{Hadzic_Lin_Rein_2021} for the gaseous tail. The turning point principle we establish is the liquid Newtonian analogue of that proven, via the separable-Hamiltonian framework, by~\cite{Lin_Zeng_2022} for gaseous stars and by~\cite{Hadzic_Lin_Rein_2021, Hadzic_Lin_2021} in the relativistic setting; our comoving formulation lets us bypass the constrained reduction and winding-index subtraction of that framework.

\section{Formulation and results}\label{S:results}

We denote by $L^2(\Omega,w)$ the weighted $L^2$ space with weight $w$, and denote by $H^k(\Omega,w)$ the weighted Sobolev space with weight $w$ --- all derivatives up to order $k$ being measured in $L^2(\Omega,w)$.

\subsection{The liquid Lane--Emden family and the mass--radius curve}

We first record the scaling structure of the family, which is the Newtonian counterpart of the central-redshift parametrisation of relativistic stars.

\begin{definition}[Liquid Lane--Emden family]\label{D:family}
For $\kappa\in(1,\infty)$ let $f_0=f_{0,\kappa}$ be the unique decreasing solution of~\eqref{E:LEODE} with $f_0(0)=\kappa$, and let $R_\kappa>0$ be defined by $f_{0,\kappa}(R_\kappa)=1$ (a proof of the existence and uniqueness of $f_{0,\kappa}$, and the fact that it decreases through the value $1$ at a finite radius, so that $R_\kappa$ exists, can be found in \cite{Lam_2024}). The liquid Lane--Emden star of central density $\kappa$ is
\begin{align}
\bar\rho_\kappa:=f_{0,\kappa}\mb 1_{[0,R_\kappa]},
\end{align}
and its mass is
\begin{align}\label{E:MASS}
M_\kappa:=4\pi\int_0^{R_\kappa}y^{d-1}\bar\rho_\kappa(y)\d y.
\end{align}
We call the planar curve
\begin{align}
\mathcal C:=\set{(R_\kappa,M_\kappa):\kappa\in(1,\infty)}\subset\R_{>0}^2
\end{align}
the \emph{mass--radius curve}, drawn with the radius on the horizontal axis.
\end{definition}

\begin{lemma}[Lane--Emden profile bound]\label{L:PROFBND}
The liquid Lane--Emden profile $\bar\rho_\kappa$ from Definition \ref{D:family} satisfies the hydrostatic identity
\begin{align}\label{E:HYDRO}
\partial_y\bar\rho_\kappa^\gamma=-\bar\rho_\kappa\,y^{1-d}m_\kappa(y)\qquad\text{on }(0,R_\kappa],\qquad\text{where}\qquad m_\kappa(y):=4\pi\int_0^y s^{d-1}\bar\rho_\kappa\,\d s.
\end{align}
In particular, we have the following bounds:
\begin{alignat}{3}
1&\le\quad\,\bar\rho_\kappa(y)&&\le\kappa,\qquad\qquad &y\in[0,R_\kappa],\label{E:RHOBOUNDS}\\
\frac{4\pi}{d}&\le-\frac{\partial_y\bar\rho_\kappa^\gamma(y)}{y}&&\le\frac{4\pi}{d}\kappa^2,\qquad\qquad &y\in(0,R_\kappa].\label{E:PROFBND}
\end{alignat}
\end{lemma}
\begin{proof}
Since $\bar\rho_\kappa$ is decreasing with $\bar\rho_\kappa(0)=\kappa$ and $\bar\rho_\kappa(R_\kappa)=1$, we get \eqref{E:RHOBOUNDS}. By integrating we can rewrite \eqref{E:LEODE} as
\[\partial_y\bar\rho_\kappa^\gamma(y)=-4\pi\bar\rho_\kappa(y)y^{1-d}\int_0^y s^{d-1}\bar\rho_\kappa(s)\d s.\]
Using this and~\eqref{E:RHOBOUNDS} we get
\begin{align*}
\frac{4\pi}{d}\le-\frac{\partial_y\bar\rho_\kappa^\gamma(y)}{y}&=\bar\rho_\kappa(y)y^{-d}4\pi\int_0^y s^{d-1}\bar\rho_\kappa\,\d s\le\frac{4\pi}{d}\kappa^2,\qquad\qquad y\in(0,R_\kappa].\qedhere
\end{align*}
\end{proof}

The gaseous Lane--Emden profiles obey a one-parameter scaling symmetry which we use throughout. Let $\bar\rho_*$ denote the gaseous Lane--Emden profile of unit central density, i.e. the solution of~\eqref{E:LEODE} with $\bar\rho_*(0)=1$, defined on its maximal interval of existence, and let
\begin{align}\label{E:MSTAR}
m_*(S):=4\pi\int_0^{S}s^{\,d-1}\bar\rho_*(s)\,\d s.
\end{align}

\begin{lemma}[Scaling of the family]\label{L:scaling}
For every $\kappa>0$ the rescaled profile $\bar\rho_\kappa(y)=\kappa\,\bar\rho_*(\kappa^{(2-\gamma)/2}y)$ solves~\eqref{E:LEODE} with central density $\kappa$. Consequently, with $S_\kappa:=\bar\rho_*^{-1}(1/\kappa)$, the liquid star of central density $\kappa>1$ has
\begin{align}\label{E:RMK}
R_\kappa=\kappa^{-{1\over 2}(2-\gamma)}\,S_\kappa,
\qquad
M_\kappa=\kappa^{{1\over 2}(d\gamma-2(d-1))}\,m_*(S_\kappa).
\end{align}
\end{lemma}
\begin{proof}
A direct computation shows that $y\mapsto\kappa\,\bar\rho_*(\kappa^{(2-\gamma)/2}y)$ solves~\eqref{E:LEODE}: the pressure term $\partial_y\bar\rho^\gamma$ picks up the factor $\kappa^{\gamma}\kappa^{(2-\gamma)/2}=\kappa^{(2+\gamma)/2}$, while the gravitational term $\bar\rho\,y^{1-d}\int_0^ys^{d-1}\bar\rho\,\d s$ picks up $\kappa^{2}\kappa^{-(2-\gamma)/2}=\kappa^{(2+\gamma)/2}$ as well, so the two scale identically; see also the self-similarity discussion in~\cite{Lam_2024}. Being decreasing with central value $\kappa$, the rescaled profile coincides with the unique profile $f_{0,\kappa}$ of Definition~\ref{D:family}, i.e. $\bar\rho_\kappa(y)=\kappa\,\bar\rho_*(\kappa^{(2-\gamma)/2}y)$. The liquid radius is then determined by $\bar\rho_\kappa(R_\kappa)=1$, i.e. $\bar\rho_*(\kappa^{(2-\gamma)/2}R_\kappa)=1/\kappa$, which is the first identity of~\eqref{E:RMK} with $S_\kappa=\bar\rho_*^{-1}(1/\kappa)$. For the mass, substituting $s=\kappa^{(2-\gamma)/2}y$ in~\eqref{E:MASS},
\begin{align*}
M_\kappa&=4\pi\int_0^{R_\kappa}y^{d-1}\kappa\bar\rho_*(\kappa^{(2-\gamma)/2}y)\d y
=\kappa^{1-\frac d2(2-\gamma)}4\pi\int_0^{S_\kappa}s^{d-1}\bar\rho_*(s)\d s\\
&=\kappa^{\frac12(d\gamma-2(d-1))}m_*(S_\kappa).\qedhere
\end{align*}
\end{proof}

The exponent of $\kappa$ in $M_\kappa$ vanishes precisely at $\gamma=\gamma_*=2(d-1)/d$, which is the origin of the threshold in the stability theory: for $\gamma>\gamma_*$ the prefactor grows in $\kappa$, for $\gamma<\gamma_*$ it decays, and the competition with the tail factor $m_*(S_\kappa)$ is what creates the extrema of $\mathcal C$.

\subsection{The linearised radial operator and the growing-mode count}

We recall the radial linearisation in comoving coordinates from~\cite{Lam_2024}. Writing the spherically symmetric flow map as $\eta(y,t)=y(1+\zeta(y,t))$ on the fixed reference interval $[0,R_\kappa]$, the linearised momentum equation for a normal mode $\zeta(y,t)=e^{\lambda t}\chi(y)$ is the singular Sturm--Liouville eigenvalue problem
\begin{subequations}\label{E:SLP}
\begin{align}
\L_\kappa\chi&=-\lambda^2\chi
\qquad\text{on}\quad[0,R_\kappa]\\
d\chi(R_\kappa)+R_\kappa\,\partial_y\chi(R_\kappa)&=0,\label{E:Robin}
\end{align}
\end{subequations}
where
\begin{align}\label{E:LDEF}
\L_\kappa\chi:=-{\gamma\over y^{d+1}\bar\rho_\kappa}\partial_y\brac{\bar\rho_\kappa^{\gamma}\,y^{d+1}\partial_y\chi}+d(\gamma_*-\gamma){\partial_y\bar\rho_\kappa^{\gamma}\over y\bar\rho_\kappa}\chi.
\end{align}
The Robin condition in~\eqref{E:SLP} encodes the liquid free boundary. We define the following function spaces
\begin{align}
  L_\kappa &:= L^2([0,R_\kappa],y^{d+1}\bar\rho_\kappa)\label{D:base space}\\
  H_\kappa &:= H^1([0,R_\kappa],y^{d+1}\bar\rho_\kappa)\label{D:form domain}
\end{align}
with the following norms
\begin{align*}
  \|\ph\|_{L_\kappa} &:= \|\ph\|_{L^2([0,R_\kappa],y^{d+1}\bar\rho_\kappa)} \sim \|\ph\|_{L^2([0,R_\kappa],y^{d+1})}\\
  \|\ph\|_{H_\kappa} &:= \|\ph\|_{H^1([0,R_\kappa],y^{d+1}\bar\rho_\kappa)} \sim \|\ph\|_{H^1([0,R_\kappa],y^{d+1})}
\end{align*}
where the equivalence in norm is due to Lemma \ref{L:PROFBND}. The operator $\L_\kappa$ is symmetric with respect to the $L_\kappa$ inner product for all sufficiently smooth functions satisfying the Robin boundary condition \eqref{E:Robin}, and its associated energy form is
\begin{align}\label{E:TE}
\<\L_\kappa\chi_1,\chi_2\>_{L_\kappa}
=\int_0^{R_\kappa} \brac{\gamma\bar\rho_\kappa^{\gamma}y^{d+1}(\partial_y\chi_1)(\partial_y\chi_2) +d(\gamma_*-\gamma)y^{d}\chi_1\chi_2\partial_y\bar\rho_\kappa^{\gamma}}\d y+d\gamma R_\kappa^d\chi_1(R_\kappa)\chi_2(R_\kappa).
\end{align}
The last, boundary, term in~\eqref{E:TE} is the distinctive contribution of the liquid surface; it is absent in the gaseous problem and is responsible for the appearance of mass extrema. This energy form together with the profile bound (Lemma \ref{L:PROFBND}) allows us to define $\L_\kappa:H_\kappa\to H_\kappa^*$, where the right-hand side of \eqref{E:TE} defines $\<\L_\kappa\chi_1,\chi_2\>$ for arbitrary $\chi_1,\chi_2\in H_\kappa$, with no boundary condition imposed on the space $H_\kappa$ --- the Robin condition will reappear as the \emph{natural} boundary condition of this form (Remark~\ref{R:weakstrong}).

\begin{definition}[Growing modes and Morse index]\label{D:count}
A \emph{growing mode} of the liquid Lane--Emden star of central density $\kappa$ is a solution $\zeta(y,t)=e^{\lambda t}\chi(y)$ of the linearised system with $\lambda>0$ and $\chi\not\equiv 0$ satisfying~\eqref{E:SLP}; equivalently, a non-trivial eigenfunction of the eigenvalue problem
\begin{align}\label{E:GEVP}
\L_\kappa\chi=\mu\chi,
\qquad\qquad d\chi(R_\kappa)+R_\kappa\partial_y\chi(R_\kappa)=0,
\end{align}
with $\mu=-\lambda^2<0$. We write $\nmodes(\kappa)$ for the number of growing modes, counted with multiplicity, and $\Mor(\L_\kappa)$ for the negative Morse index of the form~\eqref{E:TE}, i.e. the maximal dimension of a subspace of $H_\kappa$ functions on which $\<\L_\kappa\ph,\ph\><0$.
\end{definition}

\begin{remark}[Weak and classical formulations; real growth rates]\label{R:weakstrong}
The spectral machinery in the paper operates on the weak formulation such that in Lemma~\ref{L:discrete} and throughout, the eigenvalue problem~\eqref{E:GEVP} is understood in the weak sense naturally associated with the form~\eqref{E:TE}: $\chi\in H_\kappa$ is an eigenfunction with eigenvalue $\mu$ if $\<\L_\kappa\chi,\varphi\>=\mu\<\chi,\varphi\>_{L_\kappa}$ for every $\varphi\in H_\kappa$, no boundary condition being imposed on $H_\kappa$. The weak and the classical formulations coincide. Indeed, testing against $\varphi\in C_c^\infty((0,R_\kappa))$ shows that a weak eigenfunction satisfies the ODE of~\eqref{E:GEVP} in the sense of distributions on $(0,R_\kappa)$; since the coefficients are smooth on $(0,R_\kappa]$ (the profile being smooth --- see Lemma~\ref{L:joint smoothness}) and the leading coefficient $\gamma\bar\rho_\kappa^\gamma y^{d+1}$ is bounded away from zero on compact subsets of $(0,R_\kappa]$, one-dimensional elliptic regularity and a bootstrap give $\chi\in C^\infty((0,R_\kappa])$. Undoing the integration by parts that produced~\eqref{E:TE}, for arbitrary $\varphi\in H_\kappa$ the interior terms cancel against $\mu\<\chi,\varphi\>_{L_\kappa}$ and only the boundary contribution survives:
\begin{align*}
0=\<\L_\kappa\chi,\varphi\>-\mu\<\chi,\varphi\>_{L_\kappa}=\gamma R_\kappa^d\brac{d\chi(R_\kappa)+R_\kappa\partial_y\chi(R_\kappa)}\varphi(R_\kappa),
\end{align*}
using $\bar\rho_\kappa(R_\kappa)=1$; choosing $\varphi$ with $\varphi(R_\kappa)\neq0$ recovers the Robin condition of~\eqref{E:GEVP}, which is thus the \emph{natural} boundary condition of the form~\eqref{E:TE}. Conversely, every classical solution of~\eqref{E:GEVP} lying in $H_\kappa$ is a weak eigenfunction, by the same integration by parts. Finally, there are in fact no complex growing modes with $\Re\lambda>0$ and $\Im\lambda\not=0$. Indeed, the spectrum is real (Lemma~\ref{L:discrete}), so every solution $\zeta=e^{\lambda t}\chi$ of the linearised equation $\partial_t^2\zeta=-\L_\kappa\zeta$ with $\Re\lambda>0$ has $\lambda^2=-\mu\in\R$, which together with $\Re\lambda>0$ forces $\lambda\in(0,\infty)$. Definition~\ref{D:count} therefore captures \emph{all} exponentially growing normal modes.
\end{remark}

In the comoving formulation the continuity equation reads $fJ=f_0J_0=\bar\rho_\kappa$ (where as in \cite{Lam_2024} we have chosen the labelling gauge freedom $f_0J_0$ to be the reference profile $\bar\rho_\kappa$) on the fixed reference ball, so every admissible perturbation conserves the total mass $M_\kappa$ automatically. This is the structural reason why, in the liquid Newtonian case, the growing-mode count requires neither a mean-zero constraint nor the subtraction of a winding index, in contrast with the relativistic Einstein--Euler system of~\cite{Hadzic_Lin_Rein_2021}.

We single out one distinguished perturbation, the displacement that deforms a steady profile into its neighbour along the family. It will detect the turning points of the mass.

\begin{definition}[Marginal mode]\label{D:marginal}
The \emph{marginal mode} $\nu_\kappa$ of the liquid Lane--Emden star of central density $\kappa$ is the radial displacement amplitude carrying the profile $\bar\rho_\kappa$ into its neighbour $\bar\rho_{\kappa+\d\kappa}$ along the family. It is the solution, regular at the origin, of the linearised continuity relation
\begin{align}\label{E:MARGDEF}
\partial_\kappa\bar\rho_\kappa+\frac{1}{y^{d-1}}\partial_y\brac{y^{d}\bar\rho_\kappa\nu_\kappa}=0
\qquad\text{on}\quad[0,R_\kappa],
\end{align}
obtained by differentiating the mass-preservation constraint $fJ=\bar\rho_\kappa$ along $\kappa$ under the radial deformation $y\mapsto y(1+\epsilon\nu_\kappa(y)+o(\epsilon))$. Equivalently, integrating~\eqref{E:MARGDEF} against $y^{d-1}\,\d y$,
\begin{align}\label{E:MARGINT}
y^{d}\bar\rho_\kappa(y)\nu_\kappa(y)=-\int_0^{y}s^{d-1}\,\partial_\kappa\bar\rho_\kappa(s)\d s.
\end{align}
\end{definition}

To phrase the gain/loss rule we record, exactly as in~\cite{Hadzic_Lin_Rein_2021, Hadzic_Lin_2021}, the bending of the mass--radius curve at a regular point.

\begin{definition}[Turning index]\label{D:index}
At any $\kappa\in(1,\infty)$ at which $\partial_\kappa M_\kappa$ and $\partial_\kappa R_\kappa$ do not both vanish, define
\begin{align}\label{E:INDEX}
i_\kappa:=
\begin{cases}
1, &\quad\text{when}\quad(\partial_\kappa M_\kappa)(\partial_\kappa R_\kappa)>0 \quad\text{or}\quad \partial_\kappa M_\kappa=0,\\[3mm]
0, &\quad\text{when}\quad (\partial_\kappa M_\kappa)(\partial_\kappa R_\kappa)<0 \quad\text{or}\quad \partial_\kappa R_\kappa=0.
\end{cases}
\end{align}
\end{definition}

The excluded case $\partial_\kappa M_\kappa=\partial_\kappa R_\kappa=0$ in fact never occurs along the family (Lemma~\ref{L:jump}), so the turning index is defined at every $\kappa\in(1,\infty)$. It changes value as the mass--radius curve $\mathcal C$ rounds a turning point, and its jumps record the bending orientation there.

\subsection{The turning point principle}

We now state the theorems to be proven.

\begin{theorem}[Morse count, discreteness and the marginal mode]\label{T:COUNT}
Let $d\ge3$ and $\gamma\in[1,2)$, and consider the family of Definition~\ref{D:family}. For every $\kappa\in(1,\infty)$ the following hold.
\begin{enumerate}
\item The eigenvalue problem~\eqref{E:GEVP} has real eigenvalues of finite multiplicity, bounded below and accumulating only at $+\infty$. Consequently the number of growing modes is finite and equals the negative Morse index,
\begin{align}\label{E:COUNT}
\nmodes(\kappa)=\Mor(\L_\kappa)<\infty.
\end{align}
\item (Marginal mode). The marginal mode $\nu_\kappa$ of Definition~\ref{D:marginal} lies in $H_\kappa$ and satisfies the identity
\begin{align}\label{E:MARGINAL}
\<\L_\kappa\nu_\kappa,\chi\>=c_\kappa(\partial_\kappa M_\kappa)\,\chi(R_\kappa)
\quad\text{for every }\chi\in H_\kappa,
\qquad\text{where}\qquad
c_\kappa:=-\frac{M_\kappa}{4\pi R_\kappa^{d-2}}\neq0;
\end{align}
equivalently, as an element of $H_\kappa^*$,
\begin{align}\label{E:MARGDIRAC}
\L_\kappa\nu_\kappa={c_\kappa\partial_\kappa M_\kappa\over R_\kappa^{d+1}}\delta_{R_\kappa},
\end{align}
where $\delta_{R_\kappa}$ denotes the Dirac measure at the liquid surface, paired with $\chi\in H_\kappa$ through the duality $\<\delta_{R_\kappa},\chi\>_{L_\kappa}=R_\kappa^{d+1}\chi(R_\kappa)$. Moreover, we have
\begin{align}\label{E:KERNEL}
\dim\ker\L_\kappa=
\begin{cases}
1, & \quad\text{when}\quad\partial_\kappa M_\kappa=0,\\[2mm]
0, & \quad\text{when}\quad\partial_\kappa M_\kappa\neq0,
\end{cases}
\end{align}
so that $0$ is an eigenvalue of~\eqref{E:GEVP} only at critical points of $\kappa\mapsto M_\kappa$.
\end{enumerate}
\end{theorem}

\begin{theorem}[Turning point principle]\label{T:TPP}
Let $d\ge3$ and $\gamma\in[1,2)$, and consider the family $\kappa\mapsto(\bar\rho_\kappa,R_\kappa)$ of Definition~\ref{D:family}.
\begin{enumerate}
\item \emph{(Static criterion).} The number of growing modes $\nmodes(\kappa)$ is constant on every interval of $(1,\infty)$ free of critical points of $\kappa\mapsto M_\kappa$, and can change only at such critical points. At any critical point $\kappa_0$ with $\partial_\kappa^2 M_\kappa|_{\kappa_0}\neq0$, as $\kappa$ increases through $\kappa_0$ the count $\nmodes(\kappa)$ increases by $1$ if the sign of $(\partial_\kappa M_\kappa)(\partial_\kappa R_\kappa)$ changes from $-$ to $+$, and decreases by $1$ if it changes from $+$ to $-$. Equivalently, as $\kappa$ increases through any \emph{mass} critical point, $\nmodes$ jumps by the jump of the turning index $i_\kappa$ of Definition~\ref{D:index}.

\item \emph{(Threshold dichotomy and onset).} 
\begin{enumerate}
\item If $\gamma\ge\gamma_*$, then $\kappa\mapsto M_\kappa$ is strictly monotone, $\mathcal C$ has no turning point, and $\nmodes(\kappa)\equiv0$: the liquid Lane--Emden stars are linearly stable for every central density.
\item If $\gamma<\gamma_*$ and, in addition, $\gamma\ge\gamma_\sharp$ or $\mathcal D(\gamma,d)<0$ (always true for $\gamma\le\gamma_\sharp$ when $d<10$; $\mathcal D$ as in~\eqref{E:DISC}), then $\kappa\mapsto M_\kappa$ attains a strict interior local maximum. More precisely, let $\kappa_1$ be the least mass critical point at which $\partial_\kappa M_\kappa$ changes sign; then $\kappa_1$ is a strict local maximum of the mass, the radius is strictly decreasing there, $\partial_\kappa R_\kappa|_{\kappa_1}<0$, and the first growing mode appears as $\kappa$ increases through $\kappa_1$:
\begin{align}\label{E:ONSET}
\nmodes(\kappa)=0 \ \text{ for }\ 1<\kappa\le\kappa_1,\qquad\qquad \nmodes(\kappa)=1 \ \text{ for }\ \kappa_1<\kappa<\kappa_1',
\end{align}
where $\kappa_1'\in(\kappa_1,\infty]$ denotes the mass critical point following $\kappa_1$; moreover $\nmodes(\kappa)\ge1$ for all sufficiently large $\kappa$.
\end{enumerate}

\item \emph{(Large central density and the mass--radius spiral).} Suppose $\gamma<\gamma_*$.
\begin{enumerate}
\item If $\gamma<\gamma_\sharp$ and $\mathcal D(\gamma,d)<0$ (always true if $d<10$; $\mathcal D$ as in~\eqref{E:DISC}), then the rest point $\mb v^*$ governing the gaseous tail (see Proposition~\ref{P:geometry}) is a stable focus, the mass--radius curve $\mathcal C$ spirals as $\kappa\to\infty$ into the point $(R_\infty,M_\infty)$ of~\eqref{E:LIMITPOINT} representing the singular liquid star, it has infinitely many turning points, and
\begin{align}\label{E:INFTY}
\lim_{\kappa\to\infty}\nmodes(\kappa)=\infty.
\end{align}
\item If $\gamma<\gamma_\sharp$ and $\mathcal D(\gamma,d)\ge0$ (possible only when $d\ge10$), then $\mb v^*$ is a stable node, $\mathcal C$ converges to $(R_\infty,M_\infty)$ with an asymptotic tangent direction, it has finitely many turning points, and $\nmodes(\kappa)$ stabilises to a finite value as $\kappa\to\infty$; see Remark~\ref{R:open} for what remains open in this case.
\item If $\gamma_\sharp\le\gamma<\gamma_*$, then $\mathcal C$ is a single arc leaving the origin as $\kappa\searrow1$ and returning to it as $\kappa\to\infty$, with finitely many turning points, and $\nmodes(\kappa)$ stabilises to a finite \emph{odd} value as $\kappa\to\infty$; in particular the liquid stars are linearly unstable at all sufficiently large central densities. At $\gamma=\gamma_\sharp$ the family is explicit, the mass has exactly one critical point, a strict maximum at $\kappa_1=(2(d-1)/(d-2))^{(d+2)/2}$, and the stabilised value is $1$.
\end{enumerate}
\end{enumerate}
\end{theorem}

\begin{remark}
For part~(i) we stress that $\nmodes(\kappa)$ only jumps with $i_\kappa$ at mass critical points --- the turning index itself also jumps across every radius extremum with $\partial_\kappa M_\kappa\neq0$ (infinitely many of which occur along the spiral of part~(iii)(a)), where $\nmodes$ does not change. The global statement is the winding formula~\eqref{E:WINDING} of Lemma~\ref{L:winding}. Geometrically, a growing mode is gained when $\mathcal C$ bends counter-clockwise at the mass extremum, and lost when it bends clockwise. This is the precise Newtonian, liquid analogue of the turning point principle of Zel'dovich and Wheeler~\cite{Zeldovich_1963, Harrison_Thorne_Wakano_Wheeler_1965}, proven for the Einstein--Euler system in~\cite{Hadzic_Lin_Rein_2021, Hadzic_Lin_2021}. The bending product $(\partial_\kappa M_\kappa)(\partial_\kappa R_\kappa)$ is used here in place of the relativistic pairing $(\partial_\kappa M_\kappa)\partial_\kappa (M_\kappa/R_\kappa)$; the two agree in their bookkeeping of gain and loss at a mass extremum. The orientation at the onset is forced by the principle itself: since the count can never become negative, the radius must be strictly decreasing at the first mass maximum, $\partial_\kappa R_\kappa|_{\kappa_1}<0$, so the curve bends counter-clockwise there and the mass maximum produces a \emph{gained} growing mode, consistent with instability at large central density.
\end{remark}

\begin{remark}
The spiral alternative in part~(iii) is the liquid Newtonian counterpart of the relativistic mass--radius spiral (as studied by Makino, Nilsson, Uggla, Heinzle and R\"ohr~\cite{Makino_2000, Nilsson_Uggla_2000, Heinzle_Rohr_Uggla_2003}), along which the number of growing modes of the Einstein--Euler steady states tends to infinity~\cite{Hadzic_Lin_Rein_2021}. Here the spiralling is dictated by whether the gaseous Lane--Emden tail approaches its self-similar singular profile in an oscillatory (focus) or monotone (node) manner, which is precisely the sign of the discriminant~\eqref{E:DISC}. For $d<10$ the discriminant is negative throughout $\gamma\in[1,\gamma_\sharp]$, so in physical dimensions the spiral is the rule in the infinite-support regime; the node alternative of part~(iii)(b) can occur only in dimensions $d\ge10$, a threshold reminiscent of the classical role of dimension ten in the supercritical theory of the Lane--Emden equation, where the regular radial solution oscillates around the singular one precisely for $d<10$.
\end{remark}

\begin{remark}[The open node case]\label{R:open}
In the node case of part~(iii)(b) the theorem leaves the stabilised value of $\nmodes$ undetermined: for $d\ge10$, $\gamma<\gamma_\sharp$ and $\mathcal D(\gamma,d)\ge0$ we do not determine whether the liquid stars are unstable at large central density, and not even the existence of a mass maximum is guaranteed there. What the machinery does yield is a parity constraint. Since $M_\kappa\to M_\infty>0$ while $M'>0$ near $\kappa=1$ (Lemma~\ref{L:endpoints}(i)) and $M'$ has a constant sign for all large $\kappa$ (proof of part~(iii)(b)), the number of sign-changing mass critical points is odd precisely when that eventual sign is negative, i.e.\ when the mass approaches its limit $M_\infty$ from above; in that case the parity argument of part~(iii)(c) shows that the stabilised value of $\nmodes$ is odd --- in particular $\ge1$, so that the stars are linearly unstable at all sufficiently large central densities --- while in the opposite case the stabilised value is even, and could in principle be $0$. Which of the two occurs is determined by the sign of the derivative of the mass observable along the attracting eigendirection of the node, $\grad F(\mb v^*)\cdot\mb e_0$ in the notation of~\eqref{E:MASSALONG}, which we do not compute here. Parts~(ii)(b) and~(iii)(a), (c) resolve every other regime, so (iii)(b) is the only case in which the large central density behaviour of $\nmodes$ remains open.
\end{remark}

\begin{remark}[General liquid equations of state]\label{R:generalEOS}
We have studied in this paper the liquid version of the classical polytropic equation of state \eqref{E:EOS}. But our method here can be used to prove the general turning point principle for a generalised \emph{liquid} equation of state (in the spirit of what is done for the gaseous case in~\cite{Lin_Zeng_2022}):
\begin{align*}
p=P(\rho),\qquad P(1)=0,\qquad P'>0\ \text{ on }[1,\infty),
\end{align*}
say with $P\in C^2([1,\infty))$, the normalisation $P(1)=0$ fixing the surface density at $1$ as in~\eqref{E:EOS}. Depending on $P$, the shape of the exact mass--radius curve will differ, but in the case of asymptotically polytropic equations of state in the spirit of~\cite{Heinzle_Uggla_2003, Lin_Zeng_2022}: $P(\rho)=c_\infty\rho^{\gamma_\infty}(1+O(\rho^{-\epsilon}))$ as $\rho\to\infty$, with $c_\infty,\epsilon>0$, the tails of the mass--radius curve would match the polytropic case considered here.
\end{remark}

As corollaries we recover and refine the existing (in)stability theory.

\begin{corollary}[Recovery and sharpening of the radial result]\label{C:radial}
Theorem~\ref{T:TPP} implies the radial (in)stability theorem for liquid Lane--Emden stars of~\cite{Lam_2024}: for $\gamma\ge\gamma_*$ the stars are linearly stable for all central densities; for $\gamma<\gamma_*$ with $\gamma\ge\gamma_\sharp$ or $\mathcal D(\gamma,d)<0$ (in particular for all $\gamma<\gamma_*$ whenever $d<10$) they are linearly stable up to the first mass turning point $\kappa_1$ --- in particular at small relative central density --- and linearly unstable immediately beyond it and at all sufficiently large central densities. For $d\ge10$ and $\gamma\in[\gamma_\sharp,\gamma_*)$ the instability at large central density is new, sharpening the large central density instability of~\cite{Lam_2024}, which was proven for $d<10$. Moreover the precise number of growing modes at any non-critical $\kappa$ equals the sum of the jumps of the turning index over the mass turning points traversed up to $\kappa$, that is, the net number of counter-clockwise horizontal crossings of the tangent of the mass--radius curve.
\end{corollary}

Combined with the non-radial analysis of~\cite{Lam_nonradial}, carried out in the physical dimension $d=3$, Theorem~\ref{T:TPP} yields a turning point criterion for non-radial linear stability.

\begin{corollary}[Non-radial stability]\label{C:nonradial}
Let $d=3$ and $\gamma\in[1,2)$. Denote by $\mb L_\kappa$ the full three-dimensional linearised Euler--Poisson operator with no symmetry assumptions around the liquid Lane--Emden star of central density $\kappa$, as constructed in~\cite{Lam_nonradial}, acting on the space $\mathbb H_\kappa$ of admissible vector-field perturbations $\bs\theta$ there. And denote by $\nmodes_{\mathrm{full}}(\kappa)$ the number of growing modes, both radial and non-radial, counted with multiplicity, of the liquid Lane--Emden star of central density $\kappa$. All growing modes are in fact radial and we have $\nmodes_{\mathrm{full}}(\kappa)=\Mor(\mb L_\kappa)=\Mor(\L_\kappa)=\nmodes(\kappa)$. In particular, the turning point principle of Theorem \ref{T:TPP} holds for $\nmodes_{\mathrm{full}}(\kappa)$.
\end{corollary}

\begin{remark}\label{R:nonradial}
The case with the irrotational restriction in \cite{Lam_nonradial} plays no role in this equivalence; it is needed there only to upgrade the non-negativity of $\mb L_\kappa$ to strict coercivity, which excludes the solutions growing linearly in time generated by the infinite-dimensional rotational kernel $\{\grad\cdot(\bar\rho\bs\theta)=0\}$ of $\mb L_\kappa$ and thereby secures boundedness of the Lagrangian perturbation; spectral stability in the sense of the absence of exponentially growing modes, which is what the corollary addresses, is unaffected by this kernel.
\end{remark}

\section{Proof of results}\label{S:proof}

The proof is organised along the two main theorems. 

Section~\ref{S:T:COUNT} proves Theorem~\ref{T:COUNT} through two main lemmas: a spectral lemma (Lemma~\ref{L:discrete}) endowing $\L_\kappa$ with a compact resolvent and discrete spectrum, so that the growing-mode count equals the negative Morse index, $\nmodes(\kappa)=\Mor(\L_\kappa)$; and a kernel lemma (Lemma~\ref{L:kernel}) exhibiting the marginal mode $\nu_\kappa$ as the unique kernel direction, present exactly at the turning points of the mass.

Section~\ref{S:T:TPP} proves Theorem~\ref{T:TPP}. A jump lemma (Lemma~\ref{L:jump}) shows that a single eigenvalue of~\eqref{E:GEVP} crosses zero as $\kappa$ passes each mass critical point, so that $\Mor(\L_\kappa)$ changes by $\pm1$ according to the bending orientation of $\mathcal C$, and a winding lemma (Lemma~\ref{L:winding}) recasts the count globally as the winding of the tangent of $\mathcal C$. The global geometry of $\mathcal C$ is then read off from the planar dynamical system~\eqref{E:DS} governing the gaseous tail: Lemma~\ref{L:orbit} realises $\mathcal C$ as one of its orbits, Proposition~\ref{P:geometry} extracts the focus/node dichotomy and the mass--radius spiral from the sign of the discriminant~\eqref{E:DISC}, and Lemma~\ref{L:endpoints} records the behaviour at the two ends of the family. We then assemble the proofs of Theorems~\ref{T:COUNT} and~\ref{T:TPP} and deduce Corollaries~\ref{C:radial} and~\ref{C:nonradial}.

\subsection{Proof of Theorem~\ref{T:COUNT}}\label{S:T:COUNT}

The first step is the spectral structure of $\L_\kappa$, giving Theorem~\ref{T:COUNT}(i).

\begin{lemma}\label{L:discrete}
Let $\kappa\in(1,\infty)$, $\gamma\in[1,2)$ and $d\ge 3$. 
\begin{enumerate}[(i)]
  \item There exists $\alpha=\alpha(\kappa,\gamma,d)>0$ such that $(\L_\kappa+\alpha)^{-1}$, viewed as an operator on the Hilbert space $L_\kappa$, is compact, self-adjoint and positive-definite.
  \item The spectrum of $\L_\kappa$ consists of countably many real-valued eigenvalues which we can list as $\{\mu_n\}_{n=1}^\infty$ (counting multiplicity) such that the sequence is non-decreasing and tends to infinity as $n\to\infty$. Moreover, part (i) of Theorem \ref{T:COUNT} holds.
\end{enumerate}
\end{lemma}
\begin{proof}
We first prove part (i). Define the bilinear form $B_\alpha:H_\kappa\times H_\kappa\to\R$ by
\[B_\alpha[\chi_1,\chi_2] =\<(\L_\kappa+\alpha)\chi_1,\chi_2\> =\<\L_\kappa\chi_1,\chi_2\> +\alpha\<\chi_1,\chi_2\>_{L_\kappa}\]
where $\<\L_\kappa\chi_1,\chi_2\>$ is interpreted via \eqref{E:TE}.

The middle, potential, term of \eqref{E:TE} is the only one without a sign; estimating it by the profile bound of Lemma \ref{L:PROFBND} together with $\bar\rho_\kappa\ge1$, we get the lower bound
\begin{align*}
\<\L_\kappa\chi,\chi\>\ge\gamma\int_0^{R_\kappa}\bar\rho_\kappa^{\gamma}y^{d+1}(\partial_y\chi)^2\,\d y+d\gamma R_\kappa^d\chi(R_\kappa)^2-4\pi\abs{\gamma_*-\gamma}\kappa^2\norm{\chi}_{L_\kappa}^2.
\end{align*}
So for $\alpha\ge1+4\pi\abs{\gamma_*-\gamma}\kappa^2$, we have
\begin{align}
B_\alpha[\chi,\chi]\gtrsim_{d,\kappa,\gamma}\|\partial_y\chi\|_{L^2([0,R_\kappa],y^{d+1})}^2+\|\chi\|_{L_\kappa}^2\gtrsim_{d,\kappa,\gamma}\|\chi\|_{H_\kappa}^2.\label{E:bilinear form coercivity}
\end{align}
Similar bounding gives
\[|B_\alpha[\chi_1,\chi_2]| \lesssim_{d,\kappa,\gamma} \|\chi_1\|_{H_\kappa} \|\chi_2\|_{H_\kappa}.\]
Thus $B_\alpha$ is bounded and coercive, and so by the Lax--Milgram theorem, for every $g\in H_\kappa^*$ there exists a unique $\chi_g=:(\L_\kappa+\alpha)^{-1}g\in H_\kappa$ such that $B_\alpha[\chi_g,\varphi]=\<g,\varphi\>$ for all $\varphi\in H_\kappa$, and
\begin{align*}
\|g\|_{H_\kappa^*} &=\sup_{\|\chi\|_{H_\kappa}=1}|B_\alpha[\chi_g,\chi]|\lesssim_{d,\kappa,\gamma}\|\chi_g\|_{H_\kappa}\\
\|\chi_g\|_{H_\kappa}^2 &\lesssim_{d,\kappa,\gamma} B_\alpha[\chi_g,\chi_g] \leq \|g\|_{H_\kappa^*}\|\chi_g\|_{H_\kappa}.
\end{align*}
This shows that
\[\norm{(\L_\kappa+\alpha)^{-1}}_{H_\kappa^*\to H_\kappa}\lesssim_{d,\kappa,\gamma}1.\]
Every $g\in L_\kappa$ defines an element of $H_\kappa^*$ through $\varphi\mapsto\<g,\varphi\>_{L_\kappa}$, and we henceforth regard $(\L_\kappa+\alpha)^{-1}$ as an operator on the Hilbert space $L_\kappa$ by restriction; note that $H_\kappa^*$ itself carries no canonical inner product, which is why we work on $L_\kappa$. By the Rellich--Kondrachov theorem, $H^1(B_R(\R^{d+2}))$ embeds compactly into $L^2(B_R(\R^{d+2}))$. Since $H_\kappa$ and $H^1([0,R_\kappa],y^{d+1})$ have equivalent norms and the latter corresponds to the space of spherically symmetric functions in $H^1(B_{R_\kappa}(\R^{d+2}))$, we have that $H_\kappa$ embeds compactly into $L_\kappa$. Therefore $(\L_\kappa+\alpha)^{-1}:L_\kappa\to L_\kappa$, being the composition of the bounded map $L_\kappa\to H_\kappa$ with this compact embedding, is compact.

For $g_1,g_2\in L_\kappa$ we have, by the symmetry of $B_\alpha$,
\begin{align*}
\<g_1,(\L_\kappa+\alpha)^{-1}g_2\>_{L_\kappa}
&=\<g_1,(\L_\kappa+\alpha)^{-1}g_2\>
=B_\alpha[(\L_\kappa+\alpha)^{-1}g_1,(\L_\kappa+\alpha)^{-1}g_2]\\
&=B_\alpha[(\L_\kappa+\alpha)^{-1}g_2,(\L_\kappa+\alpha)^{-1}g_1]
=\<g_2,(\L_\kappa+\alpha)^{-1}g_1\>\\
&=\<g_2,(\L_\kappa+\alpha)^{-1}g_1\>_{L_\kappa};
\end{align*}
being bounded and symmetric on the Hilbert space $L_\kappa$, the operator $(\L_\kappa+\alpha)^{-1}:L_\kappa\to L_\kappa$ is self-adjoint.

For $g\in L_\kappa$ we have from \eqref{E:bilinear form coercivity} that
\begin{align*}
\<g,(\L_\kappa+\alpha)^{-1}g\>_{L_\kappa}
&=B_\alpha[(\L_\kappa+\alpha)^{-1}g,(\L_\kappa+\alpha)^{-1}g]\\
&\gtrsim_{d,\kappa,\gamma}\|(\L_\kappa+\alpha)^{-1}g\|_{H_\kappa}^2\ge0.
\end{align*}
And if $(\L_\kappa+\alpha)^{-1}g=0$, then $\<g,\varphi\>_{L_\kappa}=B_\alpha[0,\varphi]=0$ for all $\varphi\in H_\kappa$, whence $g=0$ by the density of $H_\kappa$ in $L_\kappa$. Therefore $(\L_\kappa+\alpha)^{-1}:L_\kappa\to L_\kappa$ is positive-definite.

Now we prove part (ii). We first record the correspondence between the spectra. If $\chi\in L_\kappa\setminus\set0$ satisfies $(\L_\kappa+\alpha)^{-1}\chi=\beta\chi$, then $\beta>0$ by positive-definiteness, $\chi=\beta^{-1}(\L_\kappa+\alpha)^{-1}\chi\in H_\kappa$, and unwinding the definition of $(\L_\kappa+\alpha)^{-1}$ gives $\<\L_\kappa\chi,\varphi\>=(\beta^{-1}-\alpha)\<\chi,\varphi\>_{L_\kappa}$ for all $\varphi\in H_\kappa$: $\chi$ is a (weak) eigenfunction of~\eqref{E:GEVP} with eigenvalue $\mu=\beta^{-1}-\alpha$. Conversely, a weak eigenfunction $\chi$ of~\eqref{E:GEVP} with eigenvalue $\mu$ satisfies $B_\alpha[\chi,\varphi]=(\mu+\alpha)\<\chi,\varphi\>_{L_\kappa}$ for all $\varphi\in H_\kappa$; taking $\varphi=\chi$ and using~\eqref{E:bilinear form coercivity} shows $\mu+\alpha>0$, and then $(\L_\kappa+\alpha)^{-1}\chi=(\mu+\alpha)^{-1}\chi$. The eigenvalues of~\eqref{E:GEVP} and those of $(\L_\kappa+\alpha)^{-1}$ are thus in the bijection $\mu=\beta^{-1}-\alpha$, with equal multiplicities.

Write $T:=(\L_\kappa+\alpha)^{-1}:L_\kappa\to L_\kappa$. Since $T$ is compact and self-adjoint, by the spectral theorem for compact self-adjoint operators its non-zero spectrum consists of countably many real-valued eigenvalues (each with finite multiplicity) that converge to zero, and the corresponding eigenvectors can be made into an orthonormal basis $(e_n)_{n=1}^\infty$ of $L_\kappa$; no eigenvector belongs to the spectral point $0$, as $T$ is injective by positive-definiteness, so the basis is exhausted by the non-zero eigenvalues. These are all positive, again by positive-definiteness, and we list them (with multiplicity) as a non-increasing sequence $(\beta_n)_{n=1}^\infty$ tending to $0$, with $Te_n=\beta_ne_n$.

We now identify the spectrum of $\L_\kappa$, by which we mean the spectrum of its realisation as an unbounded operator on $L_\kappa$,
\begin{align*}
A:=T^{-1}-\alpha,\qquad D(A):=TL_\kappa,
\end{align*}
the Friedrichs extension associated with the form~\eqref{E:TE}. The domain is dense, since $(TL_\kappa)^\perp=\ker T=\set0$ by the self-adjointness and injectivity of $T$; and $T^{-1}$, being the inverse of an injective self-adjoint operator with dense range, is itself self-adjoint, whence so is $A$. From $Te_n=\beta_ne_n$ we get $e_n=T(\beta_n^{-1}e_n)\in D(A)$ and $Ae_n=\mu_ne_n$, where $\mu_n:=\beta_n^{-1}-\alpha$ is a non-decreasing sequence of eigenvalues of finite multiplicity, bounded below by $-\alpha$ and tending to $+\infty$. These exhaust the spectrum of $A$. Indeed, let $\lambda\notin\set{\mu_n}_{n=1}^\infty$; then $\delta:=\inf_n\abs{\mu_n-\lambda}>0$ because $\mu_n\to\infty$, and by the completeness of the eigenbasis $(e_n)$ the operator
\begin{align*}
R_\lambda:=\sum_{n=1}^\infty(\mu_n-\lambda)^{-1}\<\,\cdot\,,e_n\>_{L_\kappa}e_n
\end{align*}
is bounded on $L_\kappa$ with norm at most $\delta^{-1}$, maps into $D(A)$, and inverts $A-\lambda$ (a routine check using the basis expansion), so $\lambda$ lies in the resolvent set of $A$. Thus $A$ has no continuous spectrum ($\mu_n$ accumulates only at $+\infty$) and, being self-adjoint, no residual spectrum. In terms of the compact resolvent, $\sigma(T)=\set{\beta_n}_{n=1}^\infty\cup\set0$, and the spectral point $0$ of $T$ --- present exactly because $A$ is unbounded --- is the image of no finite spectral point of $A$. Therefore the spectrum of $\L_\kappa$ consists precisely of the discrete eigenvalues $\mu_n=\beta_n^{-1}-\alpha$, non-decreasing, bounded below by $-\alpha$, and converging to infinity as $n\to\infty$.

By Definition~\ref{D:count} and the equivalence of the weak and classical formulations recorded in Remark~\ref{R:weakstrong}, the growing modes correspond exactly to the eigenvalues $\mu=-\lambda^2<0$, so the number of growing modes counted with multiplicity $\nmodes(\kappa)$ is $N:=|\{j:\mu_j<0\}|$, which is finite since $\mu_n\to\infty$; the $\mu_n$ being non-decreasing, these are $\mu_1,\dots,\mu_N$.

It remains to identify $N$ with the Morse index $\Mor(\L_\kappa)$. This is given by the Courant--Fischer min--max principle for the self-adjoint, bounded-below $A$ with compact resolvent:
\begin{align}
N&=\max\set{\dim W:\ W\subset H_\kappa,\ \<\L_\kappa\chi,\chi\><0\text{ for all }\chi\in W\setminus\set0}
=:\Mor(\L_\kappa)\label{E:Courant--Fischer 1}\\
\mu_n&=\inf_{\substack{W\subset H_\kappa\\ \dim W=n}}\ \sup_{\chi\in W\setminus\set0}\frac{\<\L_\kappa\chi,\chi\>}{\norm{\chi}_{L_\kappa}^2}.\label{E:Courant--Fischer 2}
\end{align}
To prove \eqref{E:Courant--Fischer 1} and \eqref{E:Courant--Fischer 2}, note firstly that the quadratic form of $A$ is the energy form: for $u=Tg\in D(A)$ and $v\in H_\kappa$,
\begin{align*}
\<(A+\alpha)u,v\>_{L_\kappa}=\<g,v\>_{L_\kappa}=B_\alpha[Tg,v]=B_\alpha[u,v],
\end{align*}
so in particular $\<Au,u\>_{L_\kappa}=\<\L_\kappa u,u\>$ on $D(A)$. Second, $D(A)=TL_\kappa$ is dense in $H_\kappa$ for the $B_\alpha$-norm, equivalent to the $H_\kappa$-norm by~\eqref{E:bilinear form coercivity} and the boundedness of $B_\alpha$, since $B_\alpha[Tg,v]=\<g,v\>_{L_\kappa}=0$ for every $g\in L_\kappa$ forces $v=0$: thus $A$ is the self-adjoint operator associated with the energy form on the form domain $H_\kappa$, its Friedrichs extension. Next, the eigenfunction identity $B_\alpha[e_n,\varphi]=(\mu_n+\alpha)\<e_n,\varphi\>_{L_\kappa}$, $\varphi\in H_\kappa$ (from the eigenvalue correspondence above), gives $B_\alpha[e_n,e_m]=(\mu_n+\alpha)\delta_{nm}$, so $\brac{(\mu_n+\alpha)^{-1/2}e_n}_{n=1}^\infty$ is a $B_\alpha$-orthonormal system in $H_\kappa$; it is complete there, since $B_\alpha[\chi,e_n]=0$ for all $n$ forces $\<\chi,e_n\>_{L_\kappa}=0$ for all $n$ and hence $\chi=0$. The two Parseval identities --- in $L_\kappa$ with coefficients $c_n:=\<\chi,e_n\>_{L_\kappa}$, and in $\brac{H_\kappa,B_\alpha}$ with coefficients $B_\alpha[\chi,(\mu_n+\alpha)^{-1/2}e_n]=(\mu_n+\alpha)^{1/2}c_n$ --- then give, for every $\chi\in H_\kappa$,
\begin{align*}
\norm{\chi}_{L_\kappa}^2=\sum_{n=1}^\infty c_n^2,\qquad B_\alpha[\chi,\chi]=\sum_{n=1}^\infty(\mu_n+\alpha)c_n^2,\qquad\text{hence}\qquad\<\L_\kappa\chi,\chi\>=\sum_{n=1}^\infty\mu_nc_n^2.
\end{align*}
Now we see that the form is negative definite on the $N$-dimensional $\Span\set{e_1,\dots,e_N}\subset D(A)\subset H_\kappa$, while every subspace $W\subset H_\kappa$ with $\dim W=N+1$ contains a non-zero $\chi$ that is $L_\kappa$-orthogonal to $e_1,\dots,e_N$, for which $\<\L_\kappa\chi,\chi\>=\sum_{n>N}\mu_nc_n^2\ge0$. Consequently, we get \eqref{E:Courant--Fischer 1} and, by the classical Rayleigh--Ritz argument, the Courant--Fischer characterisation \eqref{E:Courant--Fischer 2}. Therefore $\nmodes(\kappa)=\Mor(\L_\kappa)<\infty$, which is~\eqref{E:COUNT}.
\end{proof}

The second step identifies the kernel with the turning points of the mass, giving Theorem~\ref{T:COUNT}(ii). It is the Newtonian, mass-preserving form of the static criterion: a zero mode is an infinitesimal, mass-conserving deformation to a neighbouring equilibrium, and neighbouring equilibria are reached at fixed mass only where $\partial_\kappa M_\kappa=0$. But first, we need a lemma establishing the smoothness of the profile in $\kappa$ and $y$.

\begin{lemma}[Joint smoothness of the profile in $\kappa$ and $y$]\label{L:joint smoothness}
$(\kappa,y)\mapsto\bar\rho_\kappa(y)$ is jointly smooth (real-analytic) on $\{(\kappa,y):\kappa\in(1,\infty),\ 0\le y\le R_\kappa\}$.
\end{lemma}
\begin{proof}
By Lemma~\ref{L:scaling} we have
\[\bar\rho_\kappa(y)=\kappa\bar\rho_*(\kappa^{\frac12(2-\gamma)}y),\qquad\qquad R_\kappa=\kappa^{-\frac12(2-\gamma)}S_\kappa, \qquad\qquad M_\kappa=\kappa^{\frac12(d\gamma-2(d-1))}m_*(S_\kappa),\]
where $S_\kappa=\bar\rho_*^{-1}(1/\kappa)$. The gaseous profile $\bar\rho_*$ is real-analytic wherever it is positive: away from the origin, the enthalpy $w=\bar\rho_*^{\gamma-1}$ (respectively $h=\ln\bar\rho_*$ when $\gamma=1$) solves a second order ODE with real-analytic coefficients whose nonlinearity $w\mapsto w^{1/(\gamma-1)}$ (respectively $h\mapsto e^h$) is real-analytic on $\set{w>0}$, and solutions of real-analytic ODEs are real-analytic. At the centre, a regular singular point of the radial equation, the profile is smooth as well and indeed real-analytic in $y^2$: the nonlinearity is real-analytic near the central value $w(0)>0$, and a classical majorant argument applied to the integral form of the equation produces a convergent even power series near $y=0$ (the case $\gamma=1$ being identical with $e^h$ in place of $w^{1/(\gamma-1)}$); see e.g.~\cite{Hunter_2001} for the explicit centre series expansions of the polytropic and isothermal profiles. Moreover $\bar\rho_*$ is, by the gaseous analogue of~\eqref{E:PROFBND}, strictly decreasing wherever it is positive, so the analytic implicit function theorem renders $\kappa\mapsto S_\kappa$ real-analytic, and $m_*$ is real-analytic as the integral of a real-analytic integrand. Consequently $\kappa\mapsto R_\kappa$ and $\kappa\mapsto M_\kappa$ are real-analytic on $(1,\infty)$ --- in particular twice differentiable, and non-constant, since $M_\kappa>0$ while $M_\kappa=\kappa^{(d\gamma-2(d-1))/2}m_*(S_\kappa)\to0$ as $\kappa\searrow1$ because $S_\kappa\to0$ --- and, by the scaling relation displayed above, $(\kappa,y)\mapsto\bar\rho_\kappa(y)$ and $(\kappa,y)\mapsto\partial_\kappa\bar\rho_\kappa(y)$ are jointly smooth on $\set{(\kappa,y):\kappa\in(1,\infty),\ 0\le y\le R_\kappa}$, down to the centre $y=0$ and up to the surface $y=R_\kappa$.
\end{proof}

\begin{lemma}[Kernel/marginal mode lemma]\label{L:kernel}
The marginal mode $\nu_\kappa$ of Definition~\ref{D:marginal} obeys~\eqref{E:MARGINAL}. Moreover, $\ker \L_\kappa$ is non-trivial if and only if $\partial_\kappa M_\kappa=0$, and at a mass critical point the kernel is one-dimensional and spanned by $\nu_\kappa$.
\end{lemma}
\begin{proof}
For ease of reading we will drop the subscript $\kappa$, writing $\bar\rho=\bar\rho_\kappa$, $R=R_\kappa$, $M=M_\kappa$ and $m(y)=4\pi\int_0^y s^{d-1}\bar\rho\,\d s$ for the enclosed mass, so that $m(R)=M$ and $\bar\rho(R)=1$. 

By~\eqref{E:MARGINT} the marginal mode may be written as
\begin{align}\label{E:NUFORM}
\nu_\kappa=-\frac{G}{y^d\bar\rho}\qquad\text{where}\qquad G(y):=\int_0^y s^{d-1}\partial_\kappa\bar\rho\,\d s={\partial_\kappa m(y)\over 4\pi}.
\end{align}
Since $\partial_\kappa\bar\rho(0)=\partial_\kappa\kappa=1$ and the family $(\kappa,y)\mapsto\bar\rho_\kappa(y)$ is smooth down to the centre (see Lemma~\ref{L:joint smoothness}), we have $G(y)=\tfrac1d y^d+o(y^d)$, whence $\nu_\kappa(0)=-1/(d\kappa)$; as $\bar\rho\ge1$ is smooth on the closed interval $[0,R]$, $\nu_\kappa$ is smooth and bounded on $[0,R]$, so $\nu_\kappa\in H_\kappa$.

\emph{Step 1: $\L_\kappa\nu_\kappa=0$ in the interior.}
Differentiating~\eqref{E:NUFORM} and reducing every occurrence of $\partial_y\bar\rho$ by~\eqref{E:HYDRO} (via $\bar\rho^{\gamma-2}\partial_y\bar\rho=-\gamma^{-1}y^{1-d}m$) gives, with $G'=y^{d-1}\partial_\kappa\bar\rho$,
\begin{align*}
\bar\rho^\gamma y^{d+1}\partial_y\nu_\kappa
&=-y\bar\rho^{\gamma-1}G'+d\bar\rho^{\gamma-1}G-\frac1\gamma y^{2-d}mG,\\
\partial_y\brac{\bar\rho^\gamma y^{d+1}\partial_y\nu_\kappa}
&=-y\bar\rho^{\gamma-1}G''+\brac{(d-1)\bar\rho^{\gamma-1}-y(\gamma-1)\bar\rho^{\gamma-2}\partial_y\bar\rho-\frac1\gamma y^{2-d}m}G'\\
&\quad+\brac{d(\gamma-1)\bar\rho^{\gamma-2}\partial_y\bar\rho-\frac1\gamma\brac{y^{2-d}m'+(2-d)y^{1-d}m}}G\\
&=-y^d\bar\rho^{\gamma-1}\partial_y\partial_\kappa\bar\rho+\frac{\gamma-2}{\gamma}ym\partial_\kappa\bar\rho
+\frac{d(\gamma_*-\gamma)}{\gamma}y^{1-d}mG-\frac{4\pi}{\gamma}y\bar\rho G.
\end{align*}
Substituting the second line into~\eqref{E:LDEF} and simplifying the potential term of $\L_\kappa$ by~\eqref{E:HYDRO}, the two contributions proportional to $mGy^{-2d}\bar\rho^{-1}$ cancel, leaving
\begin{align}\label{E:LNU1}
\L_\kappa\nu_\kappa
=\frac{\gamma\bar\rho^{\gamma-2}\partial_y\partial_\kappa\bar\rho}{y}
-\frac{(\gamma-2)m\partial_\kappa\bar\rho}{y^d\bar\rho}
+\frac{4\pi G}{y^d}.
\end{align}
On the other hand, differentiating~\eqref{E:HYDRO} in $\kappa$ and using $\partial_\kappa m=4\pi G$ yields
\begin{align}\label{E:HYDROK}
\gamma\bar\rho^{\gamma-1}\partial_y\partial_\kappa\bar\rho
=-\gamma(\gamma-1)\bar\rho^{\gamma-2}\partial_y\bar\rho\,\partial_\kappa\bar\rho
-y^{1-d}m\partial_\kappa\bar\rho-4\pi y^{1-d}\bar\rho G.
\end{align}
Inserting~\eqref{E:HYDROK} into~\eqref{E:LNU1} to eliminate $\partial_y\partial_\kappa\bar\rho$, the terms in $G$ cancel and, collecting the two $m\partial_\kappa\bar\rho$ contributions,
\begin{align*}
\L_\kappa\nu_\kappa
=-(\gamma-1)\,\partial_\kappa\bar\rho\brac{\frac{m}{y^d\bar\rho}+\frac{\gamma\bar\rho^{\gamma-3}\partial_y\bar\rho}{y}}=0\qquad\text{on }(0,R),
\end{align*}
the big bracket vanishing identically since $\gamma\bar\rho^{\gamma-3}\partial_y\bar\rho=-y^{1-d}m/\bar\rho$ by~\eqref{E:HYDRO}.

\emph{Step 2: the boundary defect.}
Since $\bar\rho(R)=1$ for every $\kappa$, differentiating $f_{0,\kappa}(R_\kappa)=1$ (Definition~\ref{D:family}) gives $\partial_\kappa\bar\rho(R)=-\partial_y\bar\rho(R)\partial_\kappa R$, while differentiating $M=m(R)$ gives
\begin{align}\label{E:MPRIME}
\partial_\kappa M =4\pi G(R)+4\pi R^{d-1}\partial_\kappa R.
\end{align}
From~\eqref{E:NUFORM} and $\bar\rho(R)=1$ one finds
\begin{align*}
\nu_\kappa(R)=-\frac{G(R)}{R^d},\qquad
\partial_y\nu_\kappa(R)=-\frac{G'(R)}{R^d}+\frac{dG(R)}{R^{d+1}}+\frac{G(R)\partial_y\bar\rho(R)}{R^d},
\end{align*}
so the Robin combination of~\eqref{E:Robin} collapses to
\begin{align*}
d\nu_\kappa(R)+R\partial_y\nu_\kappa(R)=R^{1-d}\brac{-G'(R)+G(R)\partial_y\bar\rho(R)}.
\end{align*}
Using $G'(R)=R^{d-1}\partial_\kappa\bar\rho(R)=-R^{d-1}\partial_y\bar\rho(R)\partial_\kappa R$ and, from~\eqref{E:MPRIME}, $G(R)=(4\pi)^{-1}\partial_\kappa M-R^{d-1}\partial_\kappa R$, the terms carrying $\partial_\kappa R$ cancel and
\begin{align}\label{E:ROBINDEF}
d\nu_\kappa(R)+R\partial_y\nu_\kappa(R)=\frac{\partial_y\bar\rho(R)}{4\pi R^{d-1}}\partial_\kappa M.
\end{align}
Finally, evaluating~\eqref{E:HYDRO} at $y=R$ with $\bar\rho(R)=1$, $m(R)=M$ gives $\gamma\partial_y\bar\rho(R)=\partial_y\bar\rho^\gamma(R)=-R^{1-d}M$, i.e.
\begin{align}\label{E:RHOPRIME}
\partial_y\bar\rho(R)=-\frac{M}{\gamma R^{d-1}}<0 ,
\end{align}
which shows the defect~\eqref{E:ROBINDEF} is a \emph{nonzero} multiple of $\partial_\kappa M$.

\emph{Step 3: identity in $H_\kappa^*$ and the kernel.}
Let $\chi\in H_\kappa$. Integrating the first term of the energy form~\eqref{E:TE} by parts and discarding the interior integral by Step~1, only the surface terms survive, and~\eqref{E:ROBINDEF}--\eqref{E:RHOPRIME} give
\begin{align}\label{E:MARGWEAK}
\<\L_\kappa\nu_\kappa,\chi\>
=\gamma R^d\brac{d\nu_\kappa(R)+R\partial_y\nu_\kappa(R)}\chi(R)
=c_\kappa(\partial_\kappa M)\chi(R),\qquad c_\kappa:=-\frac{M_\kappa}{4\pi R_\kappa^{d-2}}\neq0.
\end{align}
This is precisely~\eqref{E:MARGINAL}: as an element of $H_\kappa^*$, $\L_\kappa\nu_\kappa$ is the boundary functional $c_\kappa\partial_\kappa M$ concentrated on the liquid surface $\{y=R\}$.

If $\partial_\kappa M_\kappa=0$, then~\eqref{E:MARGWEAK} gives $\<\L_\kappa\nu_\kappa,\chi\>=0$ for all $\chi\in H_\kappa$, so $\nu_\kappa\in\ker\L_\kappa$; as $\nu_\kappa\not\equiv0$, the kernel is non-trivial.

Conversely, let $\chi\in\ker\L_\kappa$. Testing against functions supported in $(0,R)$ shows $\L_\kappa\chi=0$ pointwise there, and $\chi\in H_\kappa$ is regular at the origin. Near $y=0$ one has $\bar\rho\to\kappa$ and, by~\eqref{E:PROFBND}, $\partial_y\bar\rho^\gamma=O(y)$, so the interior equation $\partial_y(\bar\rho^\gamma y^{d+1}\partial_y\chi)=d(\gamma_*-\gamma)\gamma^{-1}y^d\partial_y\bar\rho^\gamma\chi$ has $y=0$ as a regular singular point with indicial exponents $0$ and $-d$; any solution independent of the regular one behaves like $y^{-d}$ to leading order --- the exponents differing by the integer $d$, the Frobenius series may in addition carry a logarithmic term, which does not affect this leading order --- and so fails to lie in $L_\kappa=L^2([0,R],y^{d+1}\bar\rho)$, since $\int_0 y^{-2d}y^{d+1}\d y=\int_0 y^{1-d}\d y=\infty$ for $d\ge3$. Hence the solutions regular at the origin form a one-dimensional space; as $\nu_\kappa$ is one such solution by Step~1, we have $\chi=t\nu_\kappa$ for some $t\in\R$. Then~\eqref{E:MARGWEAK} gives $0=\<\L_\kappa\chi,\psi\>=tc_\kappa\partial_\kappa M_\kappa\psi(R)$ for all $\psi\in H_\kappa$; choosing $\psi(R)\neq0$ forces $t\partial_\kappa M_\kappa=0$.

Therefore $\ker\L_\kappa\neq\set0$ if and only if $\partial_\kappa M_\kappa=0$, and whenever $\partial_\kappa M_\kappa=0$ the above gives $\ker\L_\kappa=\Span\set{\nu_\kappa}$, of dimension one. This proves~\eqref{E:KERNEL}.
\end{proof}

The first main theorem is now proven.

\begin{proof}[Proof of Theorem~\ref{T:COUNT}]
Part~(i) is Lemma~\ref{L:discrete}(ii). For part~(ii), the identity~\eqref{E:MARGINAL}, with its explicit constant $c_\kappa=-M_\kappa/(4\pi R_\kappa^{d-2})\neq0$, is precisely the weak identity~\eqref{E:MARGWEAK} established in the proof of Lemma~\ref{L:kernel} (which also shows $\nu_\kappa\in H_\kappa$). The kernel dichotomy~\eqref{E:KERNEL} is the second assertion of Lemma~\ref{L:kernel}.
\end{proof}

\subsection{Proof of Theorem~\ref{T:TPP}}\label{S:T:TPP}

The third step controls the way the Morse index changes, giving Theorem~\ref{T:TPP}(i). Away from mass critical points no eigenvalue of~\eqref{E:GEVP} can cross zero, by Lemma~\ref{L:kernel}; across a mass critical point --- these are isolated, by real-analyticity of the family --- exactly one eigenvalue can cross, and whether it does, and in which direction, is read off from the bending of $\mathcal C$. Throughout we abbreviate $M':=\partial_\kappa M_\kappa$ and $R':=\partial_\kappa R_\kappa$.

\begin{lemma}[Jump lemma]\label{L:jump}
The maps $\kappa\mapsto R_\kappa$ and $\kappa\mapsto M_\kappa$ are real-analytic on $(1,\infty)$, the critical points of $\kappa\mapsto M_\kappa$ are isolated, and $R'\neq0$ at every mass critical point. The map $\kappa\mapsto\Mor(\L_\kappa)$ is constant on intervals free of mass critical points and, for every mass critical point $\kappa_0$, there is $\delta_0>0$ such that
\begin{align}\label{E:MORLOC}
\Mor(\L_\kappa)=
\begin{cases}
\Mor(\L_{\kappa_0})+1, &\quad\text{when}\quad M'(\kappa)R'(\kappa)>0,\\
\Mor(\L_{\kappa_0}), &\quad\text{when}\quad M'(\kappa)R'(\kappa)<0,
\end{cases}
\qquad 0<\abs{\kappa-\kappa_0}\le\delta_0.
\end{align}
Consequently, as $\kappa$ increases through any mass critical point $\kappa_0$ the index $\Mor(\L_\kappa)$ jumps by exactly the jump of the turning index $i_\kappa$ of Definition~\ref{D:index}: it increases by $1$ when the bending product $M'R'$ passes from $-$ to $+$, decreases by $1$ when it passes from $+$ to $-$, and is unchanged when the product does not change sign. In particular, at a non-degenerate mass critical point, one with $\partial_\kappa^2 M_\kappa|_{\kappa_0}\neq0$, the bending product changes sign and the jump is $\pm1$.
\end{lemma}
\begin{proof}
We fix a compact interval $K\subset(1,\infty)$ of central densities, and we denote by $\set{\mu_n(\kappa)}_{n=1}^\infty$ the non-decreasing sequence of eigenvalues of~\eqref{E:GEVP} provided by Lemma~\ref{L:discrete}, so that $\Mor(\L_\kappa)=\abs{\set{n:\mu_n(\kappa)<0}}$.

\emph{Step 1: reduction to a fixed interval and continuity of the eigenvalues.}
For $\kappa\in K$ we substitute $z=y/R_\kappa$, write $\tilde\chi(z):=\chi(R_\kappa z)$ for $\chi\in H_\kappa$, and introduce the rescaled profile
\begin{align*}
\sigma_\kappa(z):=\bar\rho_\kappa(R_\kappa z)=\kappa\,\bar\rho_*(S_\kappa z),\qquad z\in[0,1],
\end{align*}
which is jointly smooth on $K\times[0,1]$ by Lemma~\ref{L:joint smoothness} and satisfies $1\le\sigma_\kappa\le\kappa$ by~\eqref{E:RHOBOUNDS}. The substitution is an isomorphism of $H_\kappa$ onto the fixed space $X:=H^1([0,1],z^{d+1})$, with norm equivalence constants uniform over $K$, and it carries the energy form~\eqref{E:TE} and the $L_\kappa$ inner product into
\begin{align}\label{E:PULLBACK}
\<\L_\kappa\chi_1,\chi_2\>=R_\kappa^d\,Q_\kappa[\tilde\chi_1,\tilde\chi_2],
\qquad\qquad
\<\chi_1,\chi_2\>_{L_\kappa}=R_\kappa^{d+2}\,\ell_\kappa[\tilde\chi_1,\tilde\chi_2],
\end{align}
where
\begin{align}\label{E:QLDEF}
Q_\kappa[\tilde\chi_1,\tilde\chi_2]&:=\int_0^1\brac{\gamma\sigma_\kappa^\gamma z^{d+1}\partial_z\tilde\chi_1\partial_z\tilde\chi_2+W_\kappa z^{d+1}\tilde\chi_1\tilde\chi_2}\d z+d\gamma\tilde\chi_1(1)\tilde\chi_2(1),
\\
\ell_\kappa[\tilde\chi_1,\tilde\chi_2]&:=\int_0^1z^{d+1}\sigma_\kappa\tilde\chi_1\tilde\chi_2\,\d z,
\end{align}
and where, by the hydrostatic identity~\eqref{E:HYDRO} and the substitution $s=R_\kappa zu$ in the enclosed-mass integral,
\begin{align}\label{E:WDEF}
W_\kappa(z):=d(\gamma_*-\gamma)\frac{\partial_z\sigma_\kappa^\gamma(z)}{z}
=-4\pi d(\gamma_*-\gamma)R_\kappa^2\sigma_\kappa(z)\int_0^1u^{d-1}\sigma_\kappa(zu)\d u
\end{align}
is jointly smooth on $K\times[0,1]$. By the spectral resolution of Lemma~\ref{L:discrete} and the Courant--Fischer principle, applied as in the proof of that lemma, \eqref{E:PULLBACK} gives
\begin{align}\label{E:MINMAX}
\mu_n(\kappa)=R_\kappa^{-2}\lambda_n(\kappa),
\qquad
\lambda_n(\kappa):=\inf_{\substack{W\subset X\\ \dim W=n}}\ \sup_{\tilde\chi\in W\setminus\set0}\frac{Q_\kappa[\tilde\chi,\tilde\chi]}{\ell_\kappa[\tilde\chi,\tilde\chi]},
\end{align}
and the eigenfunctions of~\eqref{E:GEVP} correspond under the substitution to the eigenfunctions of the pencil $(Q_\kappa,\ell_\kappa)$ on $X$, i.e. to solutions of $Q_\kappa[\tilde\chi,\tilde\varphi]=\lambda\ell_\kappa[\tilde\chi,\tilde\varphi]$ for all $\tilde\varphi\in X$.

We claim each $\lambda_n$, and hence each $\mu_n$, is continuous on $K$. Set $\Lambda:=1+\sup_{K\times[0,1]}|W_\kappa|$ and $P_\kappa:=Q_\kappa+\Lambda\ell_\kappa$. Since $\sigma_\kappa\ge1$ and the boundary term of~\eqref{E:QLDEF} is non-negative,
\begin{align}\label{E:GARDING}
P_\kappa[\tilde\chi,\tilde\chi]\ge\gamma\int_0^1z^{d+1}(\partial_z\tilde\chi)^2\,\d z+\int_0^1z^{d+1}\tilde\chi^2\,\d z\ge\min(\gamma,1)\norm{\tilde\chi}_X^2\ge0.
\end{align}
The coefficients $\gamma\sigma_\kappa^\gamma$, $W_\kappa$ and $\sigma_\kappa$ are jointly continuous on the compact set $K\times[0,1]$, so for $\kappa,\kappa'\in K$ we have, with a modulus of continuity $\omega=\omega(\abs{\kappa'-\kappa})\to0$ as $\kappa'\to\kappa$,
\begin{align*}
\abs{Q_{\kappa'}[\tilde\chi,\tilde\chi]-Q_{\kappa}[\tilde\chi,\tilde\chi]}\le\omega\norm{\tilde\chi}_X^2,
\qquad
\abs{\ell_{\kappa'}[\tilde\chi,\tilde\chi]-\ell_{\kappa}[\tilde\chi,\tilde\chi]}\le\omega\,\ell_\kappa[\tilde\chi,\tilde\chi],
\end{align*}
the boundary term of $Q_\kappa$ being independent of $\kappa$ and $\sigma_\kappa\ge1$. By~\eqref{E:GARDING} these combine into the two-sided bounds
\begin{align*}
(1-\delta)P_\kappa\le P_{\kappa'}\le(1+\delta)P_\kappa,
\qquad
(1-\delta)\ell_\kappa\le\ell_{\kappa'}\le(1+\delta)\ell_\kappa,
\qquad
\delta:=\frac{(1+\Lambda)\,\omega}{\min(\gamma,1)}.
\end{align*}
Since $\lambda_n(\kappa)+\Lambda\ge0$ is the min--max value~\eqref{E:MINMAX} of the ratio $P_\kappa/\ell_\kappa$, and two-sided bounds between non-negative forms are preserved under taking $\sup$ and $\inf$,
\begin{align*}
\frac{1-\delta}{1+\delta}\brac{\lambda_n(\kappa)+\Lambda}\le\lambda_n(\kappa')+\Lambda\le\frac{1+\delta}{1-\delta}\brac{\lambda_n(\kappa)+\Lambda},
\end{align*}
and letting $\kappa'\to\kappa$ proves the claim.

\emph{Step 2: constancy away from mass critical points.}
Let $J\subset(1,\infty)$ be an interval free of critical points of $\kappa\mapsto M_\kappa$, let $\kappa\in J$ and set $N:=\Mor(\L_\kappa)$. By Lemma~\ref{L:kernel}, $\ker\L_{\kappa'}=\set0$ for every $\kappa'\in J$, i.e. no eigenvalue vanishes on $J$; in particular $\mu_N(\kappa)<0<\mu_{N+1}(\kappa)$ (the first inequality being void when $N=0$). By Step~1 both inequalities persist for $\kappa'\in J$ near $\kappa$, and then the ordering of the eigenvalues gives $\Mor(\L_{\kappa'})=N$. Hence $\kappa\mapsto\Mor(\L_\kappa)$ is locally constant on $J$, and $J$ being an interval, constant on $J$.

\emph{Step 3: mass critical points are regular points of $\mathcal C$.}
Let $\kappa_0$ be a mass critical point and suppose, for contradiction, that also $R'(\kappa_0)=0$. Then \eqref{E:MPRIME} gives $4\pi G(R)=M'-4\pi R^{d-1}R'=0$ at $\kappa_0$, whence $\nu_{\kappa_0}(R)=-G(R)/R^d=0$ by~\eqref{E:NUFORM}; and since the Robin defect~\eqref{E:ROBINDEF} vanishes when $M'=0$, also $\partial_y\nu_{\kappa_0}(R)=-d\nu_{\kappa_0}(R)/R=0$. By Step~1 of the proof of Lemma~\ref{L:kernel}, $\nu_{\kappa_0}$ solves the linear second order ODE $\partial_y(\bar\rho^\gamma y^{d+1}\partial_y\nu)=d(\gamma_*-\gamma)\gamma^{-1}y^d\partial_y\bar\rho^\gamma\nu$ on $(0,R]$, whose leading coefficient $\bar\rho^\gamma y^{d+1}$ is bounded away from zero on $[\epsilon,R]$ for every $\epsilon>0$; uniqueness for the initial value problem at $y=R$ with the trivial data thus forces $\nu_{\kappa_0}\equiv0$ on $(0,R]$, contradicting $\nu_{\kappa_0}(0)=-1/(d\kappa_0)\neq0$ and the continuity of $\nu_{\kappa_0}$ on $[0,R]$. Hence $R'\neq0$ at every mass critical point.

Moreover, since $\kappa\mapsto M_\kappa$ is real-analytic and non-constant (Lemma~\ref{L:joint smoothness}), its derivative $M'$ is real-analytic and not identically zero, so the mass critical points are isolated. Near a mass critical point $\kappa_0$ the derivative $M'$ therefore has a constant sign on each of the punctured side of $\kappa_0$, while $R'\neq0$ in a neighbourhood of $\kappa_0$; hence the bending product $M'R'$ has a constant sign on each punctured side of $\kappa_0$ and, by Definition~\ref{D:index}, for $\kappa\neq\kappa_0$ sufficiently close to $\kappa_0$,
\begin{align}\label{E:INDEXNEAR}
i_\kappa=
\begin{cases}
1, &\quad\text{when}\quad M'(\kappa)R'(\kappa)>0,\\
0, &\quad\text{when}\quad M'(\kappa)R'(\kappa)<0.
\end{cases}
\end{align}
If in addition $M''(\kappa_0)\neq0$, then $M'$, and with it the bending product, changes sign across $\kappa_0$.

\emph{Step 4: localisation of the spectrum near $\kappa_0$.}
Let $\kappa_0$ be a mass critical point and set $N:=\Mor(\L_{\kappa_0})$. Since $M'(\kappa_0)=0$, Lemma~\ref{L:kernel} makes $0$ an eigenvalue of~\eqref{E:GEVP} of multiplicity one, so
\begin{align*}
\mu_N(\kappa_0)<0=\mu_{N+1}(\kappa_0)<\mu_{N+2}(\kappa_0)
\end{align*}
(the first inequality again void when $N=0$). By Step~1 there is $\delta_0>0$ such that $\mu_N(\kappa)<0<\mu_{N+2}(\kappa)$ for $\abs{\kappa-\kappa_0}\le\delta_0$; shrinking $\delta_0$ if necessary, Step~3 (the critical point $\kappa_0$ being isolated) and Lemma~\ref{L:kernel} in addition give $\mu_{N+1}(\kappa)\neq0$ for $0<\abs{\kappa-\kappa_0}\le\delta_0$. By the ordering of the eigenvalues,
\begin{align}\label{E:MULOC}
\Mor(\L_\kappa)=
\begin{cases}
N+1, &\quad\text{when}\quad \mu_{N+1}(\kappa)<0,\\
N, &\quad\text{when}\quad \mu_{N+1}(\kappa)>0,
\end{cases}
\qquad 0<\abs{\kappa-\kappa_0}\le\delta_0,
\end{align}
and~\eqref{E:MORLOC} reduces to determining the sign of $\mu_{N+1}$ on either side of $\kappa_0$.

\emph{Step 5: the direction of the crossing.}
Write $\mu(\kappa):=\mu_{N+1}(\kappa)$, and let $\chi_\kappa\in H_\kappa$ be a corresponding eigenfunction of~\eqref{E:GEVP} whose pullback $\tilde\chi_\kappa\in X$ is normalised by $\ell_\kappa[\tilde\chi_\kappa,\tilde\chi_\kappa]=1$. We claim that, after shrinking $\delta_0$ once more,
\begin{align}\label{E:CROSSING}
\operatorname{sign}\mu(\kappa)=-\operatorname{sign}\brac{M'(\kappa)R'(\kappa)},
\qquad 0<\abs{\kappa-\kappa_0}\le\delta_0.
\end{align}
The starting point is a pairing of the eigenfunction with the marginal mode: by the symmetry of the energy form~\eqref{E:TE} and the marginal identity~\eqref{E:MARGWEAK},
\begin{align}\label{E:PAIR}
\mu(\kappa)\<\chi_\kappa,\nu_\kappa\>_{L_\kappa}
=\<\L_\kappa\chi_\kappa,\nu_\kappa\>
=\<\L_\kappa\nu_\kappa,\chi_\kappa\>
=c_\kappa M'(\kappa)\chi_\kappa(R_\kappa),
\qquad
c_\kappa=-\frac{M_\kappa}{4\pi R_\kappa^{d-2}}<0 ,
\end{align}
or, in the rescaled variables of Step~1 via~\eqref{E:PULLBACK},
\begin{align}\label{E:PAIRRESC}
\mu(\kappa)\,R_\kappa^{d+2}\,\ell_\kappa[\tilde\chi_\kappa,\tilde\nu_\kappa]=c_\kappa M'(\kappa)\,\tilde\chi_\kappa(1),
\qquad
\tilde\nu_\kappa(z):=\nu_\kappa(R_\kappa z).
\end{align}
Note that by~\eqref{E:NUFORM} and the substitution $s=R_\kappa zu$,
\begin{align}\label{E:NUPULL}
\tilde\nu_\kappa(z)=-\frac{1}{\sigma_\kappa(z)}\int_0^1u^{d-1}\,\partial_\kappa\bar\rho_\kappa(R_\kappa zu)\,\d u ,
\end{align}
which by Lemma~\ref{L:joint smoothness} is jointly continuous on $K\times[0,1]$; in particular $\tilde\nu_\kappa\to\tilde\nu_{\kappa_0}$ uniformly on $[0,1]$ as $\kappa\to\kappa_0$.

Suppose~\eqref{E:CROSSING} fails. Then there are $\kappa_j\to\kappa_0$, $\kappa_j\neq\kappa_0$, with
\begin{align}\label{E:CONTRA}
\mu(\kappa_j)M'(\kappa_j)R'(\kappa_j)>0,
\end{align}
all three factors being non-zero for $0<\abs{\kappa_j-\kappa_0}\le\delta_0$ by Steps~3--4. Abbreviate $\tilde\chi_j:=\tilde\chi_{\kappa_j}$, $\tilde\nu_j:=\tilde\nu_{\kappa_j}$ and $\lambda_j:=R_{\kappa_j}^2\mu(\kappa_j)$.

First, $(\tilde\chi_j)_j$ is bounded in $X$: by Step~1, $\mu$ is continuous with $\mu(\kappa_0)=0$, so $Q_{\kappa_j}[\tilde\chi_j,\tilde\chi_j]=\lambda_j\to0$, and~\eqref{E:GARDING} gives $\min(\gamma,1)\norm{\tilde\chi_j}_X^2\le P_{\kappa_j}[\tilde\chi_j,\tilde\chi_j]=\lambda_j+\Lambda\le\Lambda+1$ for $j$ large. Passing to a subsequence, $\tilde\chi_j\rightharpoonup\tilde\chi_*$ weakly in $X$; by the compact embedding of $X$ into $L^2([0,1],z^{d+1})$ (as in the proof of Lemma~\ref{L:discrete}) the convergence is strong in $L^2([0,1],z^{d+1})$, and since the point evaluation $\tilde\chi\mapsto\tilde\chi(1)$ is a bounded linear functional on $X$ (a one-dimensional trace estimate on $[1/2,1]$, where $z^{d+1}\sim1$), also $\tilde\chi_j(1)\to\tilde\chi_*(1)$.

Second, $\tilde\chi_*$ is a kernel element. For fixed $\tilde\varphi\in X$ the eigenvalue equation reads $Q_{\kappa_j}[\tilde\chi_j,\tilde\varphi]=\lambda_j\,\ell_{\kappa_j}[\tilde\chi_j,\tilde\varphi]$; by the uniform convergence of the coefficients together with the boundedness of $(\tilde\chi_j)_j$ in $X$, the weak convergence $\partial_z\tilde\chi_j\rightharpoonup\partial_z\tilde\chi_*$ in $L^2([0,1],z^{d+1})$, the strong convergence of $\tilde\chi_j$, the convergence of the boundary values, and $\lambda_j\to0$, we may pass to the limit and obtain $Q_{\kappa_0}[\tilde\chi_*,\tilde\varphi]=0$ for every $\tilde\varphi\in X$. Undoing the substitution, $\<\L_{\kappa_0}\chi_*,\varphi\>=0$ for all $\varphi\in H_{\kappa_0}$, so $\chi_*\in\ker\L_{\kappa_0}=\Span\set{\nu_{\kappa_0}}$ by Lemma~\ref{L:kernel}, recalling $M'(\kappa_0)=0$. Moreover $\ell_{\kappa_0}[\tilde\chi_*,\tilde\chi_*]=\lim_{j\to\infty}\ell_{\kappa_j}[\tilde\chi_j,\tilde\chi_j]=1$ by the strong convergence, so
\begin{align*}
\tilde\chi_*=t\tilde\nu_{\kappa_0}\qquad\text{for some}\quad t\neq0.
\end{align*}

Third, we pass to the limit in~\eqref{E:PAIRRESC} along the subsequence. By the strong $L^2$ convergence of $\tilde\chi_j$ and the uniform convergence of $\sigma_{\kappa_j}$ and $\tilde\nu_j$,
\begin{align*}
\ell_{\kappa_j}[\tilde\chi_j,\tilde\nu_j]\to\ell_{\kappa_0}[\tilde\chi_*,\tilde\nu_{\kappa_0}]=t\ell_{\kappa_0}[\tilde\nu_{\kappa_0},\tilde\nu_{\kappa_0}]\neq0,
\end{align*}
while, by~\eqref{E:NUFORM} and~\eqref{E:MPRIME} with $M'(\kappa_0)=0$,
\begin{align*}
\tilde\chi_j(1)\to\tilde\chi_*(1)=t\nu_{\kappa_0}(R_{\kappa_0})=-t\frac{G(R_{\kappa_0})}{R_{\kappa_0}^d}=t\frac{R'(\kappa_0)}{R_{\kappa_0}}.
\end{align*}
Dividing~\eqref{E:PAIRRESC} by $M'(\kappa_j)\neq0$ and letting $j\to\infty$, the factor $t$ cancels and
\begin{align*}
\frac{\mu(\kappa_j)}{M'(\kappa_j)}\to\frac{c_{\kappa_0}}{R_{\kappa_0}^{d+3}\,\ell_{\kappa_0}[\tilde\nu_{\kappa_0},\tilde\nu_{\kappa_0}]}\,R'(\kappa_0).
\end{align*}
Since $c_{\kappa_0}<0$, the limit is a non-zero number of sign opposite to $R'(\kappa_0)$. Hence for $j$ large $\operatorname{sign}\mu(\kappa_j)=-\operatorname{sign}\brac{M'(\kappa_j)R'(\kappa_0)}=-\operatorname{sign}\brac{M'(\kappa_j)R'(\kappa_j)}$, contradicting~\eqref{E:CONTRA}. This proves~\eqref{E:CROSSING}.

\emph{Conclusion.} Combining~\eqref{E:MULOC} with~\eqref{E:CROSSING} yields exactly~\eqref{E:MORLOC}, with $N=\Mor(\L_{\kappa_0})$. By Step~3 the bending product $M'R'$ has a constant sign on each punctured side of $\kappa_0$. If it passes from $-$ to $+$ as $\kappa$ increases through $\kappa_0$, the index jumps from $N$ to $N+1$; if it passes from $+$ to $-$, it jumps from $N+1$ to $N$; and if it has the same sign on both sides, the index takes the same value on both sides and does not jump. Comparing with~\eqref{E:INDEXNEAR}, in every case the jump of $\Mor(\L_\kappa)$ equals the jump of the turning index $i_\kappa$; and by Step~3 the sign-changing cases occur, in particular, whenever $M''(\kappa_0)\neq0$. This completes the proof.
\end{proof}

We record next a global, winding-number form of the count, obtained by adding up the jumps of Lemma~\ref{L:jump} along the curve. It identifies the growing-mode count with the winding of the \emph{tangent} of $\mathcal C$, measured in units of $\pi$, and it is the exact Newtonian counterpart of the winding index of~\cite{Hadzic_Lin_Rein_2021, Hadzic_Lin_2021}.

\begin{lemma}[Winding form of the count]\label{L:winding}
The tangent vector $\brac{R'(\kappa),M'(\kappa)}$ of the mass--radius curve never vanishes, and therefore admits a continuous polar angle $\varphi:(1,\infty)\to\R$, unique up to an additive integer multiple of $2\pi$; the mass critical points are exactly the parameters with $\varphi\in\pi\Z$, at which the tangent of $\mathcal C$ is horizontal. For all $a<b$ in $(1,\infty)$ that are not mass critical points,
\begin{align}\label{E:WINDING}
\nmodes(b)-\nmodes(a)=\left\lfloor\frac{\varphi(b)}{\pi}\right\rfloor-\left\lfloor\frac{\varphi(a)}{\pi}\right\rfloor.
\end{align}
In particular the count increases by $1$ exactly when the tangent of $\mathcal C$ crosses the horizontal direction counter-clockwise, $\varphi$ increasing through a multiple of $\pi$, and decreases by $1$ at a clockwise crossing.
\end{lemma}
\begin{proof}
If $M'(\kappa)=0$, then $R'(\kappa)\neq0$ by Lemma~\ref{L:jump}, so $(R',M')\neq(0,0)$ everywhere, and a continuous argument $\varphi$ exists (and is unique modulo $2\pi\Z$) because $(1,\infty)$ is an interval. Writing $\varrho:=\abs{(R',M')}>0$ we have $R'=\varrho\cos\varphi$ and $M'=\varrho\sin\varphi$, so mass critical points correspond to $\varphi\in\pi\Z$.

By Lemma~\ref{L:jump} the mass critical points are isolated, hence finite in number in $[a,b]$: label them $a<\kappa^1<\dots<\kappa^m<b$. On each of the complementary subintervals of $[a,b]$ the angle $\varphi$ avoids $\pi\Z$, so $\lfloor\varphi/\pi\rfloor$ is constant there, being an integer-valued continuous function; and $\nmodes=\Mor(\L_\kappa)$ is constant there as well, by Theorem~\ref{T:COUNT} and Lemma~\ref{L:jump}. It therefore suffices to match the jumps of the two sides of~\eqref{E:WINDING} across each $\kappa^j$.

Fix $j$ and write $\varphi=k_j\pi+\psi_j$ near $\kappa^j$, where $k_j\in\Z$ and $\psi_j$ is continuous and small with $\psi_j(\kappa^j)=0$ and $\psi_j\neq0$ off $\kappa^j$. Since
\begin{align*}
M'R'=\varrho^2\sin\varphi\cos\varphi=\varrho^2\sin(k_j\pi+\psi_j)\cos(k_j\pi+\psi_j)=\varrho^2\sin\psi_j\cos\psi_j,
\end{align*}
we have $\operatorname{sign}(M'R')=\operatorname{sign}\psi_j$ on either punctured side of $\kappa^j$. Hence, by~\eqref{E:MORLOC}, near $\kappa^j$,
\begin{align*}
\nmodes(\kappa)=\Mor(\L_{\kappa^j})+\mb 1_{\set{\psi_j(\kappa)>0}},
\qquad 0<\abs{\kappa-\kappa^j}\ \text{small},
\end{align*}
while directly $\lfloor\varphi/\pi\rfloor=k_j-\mb 1_{\set{\psi_j<0}}$ there. Both sides of~\eqref{E:WINDING} therefore jump across $\kappa^j$ by the same amount, namely $\mb 1_{\set{\psi_j>0\text{ after}}}-\mb 1_{\set{\psi_j>0\text{ before}}}\in\set{-1,0,1}$. Summing over $j=1,\dots,m$ and using the constancy in between proves~\eqref{E:WINDING}. The final sentence is the case bookkeeping: $\varphi$ increasing through $k_j\pi$ means $\psi_j$ passes from $-$ to $+$, a jump of $+1$; decreasing, from $+$ to $-$, a jump of $-1$; and a touching without crossing gives no jump.
\end{proof}

The remaining input is the global geometry of the mass--radius curve, which supplies Theorem~\ref{T:TPP}(ii)--(iii). Using the substitution
\begin{align}\label{E:DSVARS}
u_1(r):=r^{\frac{2}{2-\gamma}}\bar\rho(r),\qquad u_2(r):=r^{\frac{2}{2-\gamma}-d}m(r),\qquad m(r)=4\pi\int_0^r y^{d-1}\bar\rho\,\d y,
\end{align}
and the logarithmic time $\tau=\ln r$, the gaseous Lane--Emden equation~\eqref{E:LEODE} becomes the planar autonomous system
\begin{align}\label{E:DS}
\frac{\d}{\d\tau}\mat{v_1\\v_2}=\mat{-\tfrac1\gamma v_1^{2-\gamma}v_2+\tfrac{2}{2-\gamma}v_1\\[1mm] 4\pi v_1-\brac{d-\tfrac{2}{2-\gamma}}v_2},\qquad v_j(\tau)=u_j(e^\tau).
\end{align}
Details of this can be found in \cite{Lam_2024}. From now on $\mb v=(v_1,v_2)$ denotes the orbit of~\eqref{E:DS} associated with the unit-central-density gaseous profile, i.e.~\eqref{E:DSVARS} is evaluated at $\bar\rho=\bar\rho_*$, defined for $\tau$ in
\begin{align*}
I:=
\begin{cases}
\R, &\quad\text{when}\quad \gamma\le\gamma_\sharp,\\
(-\infty,\ln S_*), &\quad\text{when}\quad \gamma>\gamma_\sharp,
\end{cases}
\end{align*}
according to the support dichotomy of~\cite{Lam_2024}: the gaseous profile $\bar\rho_*$ has infinite support when $\gamma\le\gamma_\sharp$ and compact support $[0,S_*]$ when $\gamma>\gamma_\sharp$.

The first observation is that the mass--radius curve is nothing but this orbit, viewed through a fixed change of coordinates. This is the exact mechanism by which the phase portrait of~\eqref{E:DS} is transcribed into the geometry of $\mathcal C$.

\begin{lemma}[The mass--radius curve as an orbit]\label{L:orbit}
Let $\gamma\in[1,2)$ and set
\begin{align}\label{E:PHI}
a:={2(d-1)-d\gamma\over 2},
\qquad
\Phi(v_1,v_2):=\brac{v_1^{\frac{2-\gamma}{2}},\, v_1^{a}v_2},
\end{align}
so that $\Phi$ is a real-analytic diffeomorphism of $(0,\infty)^2$ onto itself. Then $\tau_\kappa:=\ln S_\kappa$ defines a real-analytic, strictly increasing bijection of $(1,\infty)$ onto $I$, with inverse $\tau\mapsto1/\bar\rho_*(e^\tau)$, and
\begin{align}\label{E:CURVEORBIT}
(R_\kappa,M_\kappa)=\Phi(\mb v(\tau_\kappa)),\qquad\kappa\in(1,\infty).
\end{align}
That is, the mass--radius curve $\mathcal C$ is the image of the gaseous orbit under the fixed diffeomorphism $\Phi$, traversed in the direction of increasing $\tau$, with $\tau_\kappa\to\sup I$ as $\kappa\to\infty$.
\end{lemma}
\begin{proof}
The map $\Phi$ is a bijection of $(0,\infty)^2$ with the explicit real-analytic inverse
\[\Phi^{-1}(R,M)=\brac{R^{\frac{2}{2-\gamma}},\, M R^{-\frac{2a}{2-\gamma}}},\]
so it is a real-analytic diffeomorphism. For $\tau\in I$ put $S:=e^\tau\in(0,\sup e^{I})$; then $\bar\rho_*(S)\in(0,1)$, and $\kappa=1/\bar\rho_*(S)>1$ is the unique central density with $S_\kappa=S$, i.e. $\tau_\kappa=\tau$. Hence $\kappa\mapsto\tau_\kappa$ is a bijection of $(1,\infty)$ onto $I$ with inverse $\tau\mapsto1/\bar\rho_*(e^\tau)$; it is strictly increasing and real-analytic together with its inverse because $\bar\rho_*$ is real-analytic and strictly decreasing wherever positive (Lemma~\ref{L:joint smoothness}).

Next, evaluating~\eqref{E:DSVARS} at $\bar\rho=\bar\rho_*$, $r=S_\kappa$, and using $\bar\rho_*(S_\kappa)=1/\kappa$,
\begin{align*}
v_1(\tau_\kappa)=S_\kappa^{\frac{2}{2-\gamma}}\,\bar\rho_*(S_\kappa)=\frac{S_\kappa^{\frac{2}{2-\gamma}}}{\kappa},
\qquad
v_2(\tau_\kappa)=S_\kappa^{\frac{2}{2-\gamma}-d}\,m_*(S_\kappa).
\end{align*}
The first identity of~\eqref{E:RMK} then reads
\begin{align*}
R_\kappa=\kappa^{-\frac{2-\gamma}{2}}S_\kappa=\brac{\frac{S_\kappa^{\frac{2}{2-\gamma}}}{\kappa}}^{\frac{2-\gamma}{2}}=v_1(\tau_\kappa)^{\frac{2-\gamma}{2}},
\end{align*}
while the second, $M_\kappa=\kappa^{-a}m_*(S_\kappa)$, becomes
\begin{align*}
M_\kappa
=\brac{\frac{S_\kappa^{\frac{2}{2-\gamma}}}{\kappa}}^{a}S_\kappa^{-\frac{2a}{2-\gamma}}\,m_*(S_\kappa)
=v_1(\tau_\kappa)^{a}\,S_\kappa^{\,d-\frac{2a+2}{2-\gamma}}\,v_2(\tau_\kappa)
=v_1(\tau_\kappa)^{a}\,v_2(\tau_\kappa),
\end{align*}
the powers of $S_\kappa$ cancelling because $2a+2=2(d-1)-d\gamma+2=d(2-\gamma)$, so that $d-\frac{2a+2}{2-\gamma}=0$. This is~\eqref{E:CURVEORBIT}.
\end{proof}

\begin{proposition}[Mass--radius geometry from the planar system]\label{P:geometry}
Assume $2(d-1)-d\gamma>0$, i.e. $\gamma<\gamma_*$. The system~\eqref{E:DS} has the rest points $\mb 0$ and 
\begin{align}\label{E:RESTPOINT}
\mb v^*=\mat{\displaystyle\brac{\frac1{2\pi}\frac{-d\gamma^2+2(d-1)\gamma}{(2-\gamma)^2}}^{\frac1{2-\gamma}}\\[1mm] \displaystyle\frac{2\gamma}{2-\gamma}\brac{\frac1{2\pi}\frac{-d\gamma^2+2(d-1)\gamma}{(2-\gamma)^2}}^{\frac{\gamma-1}{2-\gamma}}},
\end{align}
the latter corresponding to the self-similar singular profile $\bar\rho\sim c\,r^{-2/(2-\gamma)}$. The linearisation of~\eqref{E:DS} at $\mb v^*$ has eigenvalues
\begin{align}
\lambda_\pm=\frac{2}{2-\gamma}-1-\frac d2\pm\frac12\sqrt{\mathcal D(\gamma,d)\,(2-\gamma)^{-2}},
\end{align}
with $\mathcal D$ as in~\eqref{E:DISC}: $\mb v^*$ is a focus when $\mathcal D(\gamma,d)<0$ and a node when $\mathcal D(\gamma,d)\ge0$; moreover, when $d<10$ one has $\mathcal D(\gamma,d)<0$ for every $\gamma\in[1,\gamma_\sharp]$. Suppose now $\gamma<\gamma_\sharp$. Then $\mb v^*$ is exponentially asymptotically stable, the gaseous orbit converges to it, $\mb v(\tau)\to\mb v^*$ as $\tau\to\infty$, and:
\begin{enumerate}[(i)]
\item if $\mathcal D(\gamma,d)<0$, both the position $\mb v(\tau)-\mb v^*$ and the velocity $\frac{\d}{\d\tau}\mb v(\tau)$ wind around the origin unboundedly: any continuous choice of their polar angles tends to $+\infty$ or $-\infty$ linearly in $\tau$;
\item if $\mathcal D(\gamma,d)\ge0$, the normalised velocity $\frac{\d}{\d\tau}\mb v(\tau)/|\frac{\d}{\d\tau}\mb v(\tau)|$ converges, as $\tau\to\infty$, to a unit eigenvector of the linearisation.
\end{enumerate}
Consequently, by Lemma~\ref{L:orbit}, as $\kappa\to\infty$ the mass--radius curve $\mathcal C$ converges to the point
\begin{align}\label{E:LIMITPOINT}
(R_\infty,M_\infty):=\Phi(\mb v^*)=\brac{(v_1^*)^{\frac{2-\gamma}{2}},\ (v_1^*)^{a}\,v_2^*},
\end{align}
which is precisely the radius and mass of the singular liquid star obtained by truncating the singular profile $v_1^*\,r^{-2/(2-\gamma)}$ at density $1$; in case~(i) the curve $\mathcal C$ spirals into $(R_\infty,M_\infty)$, any continuous polar angle of its tangent $\brac{R'(\kappa),M'(\kappa)}$ tending to $+\infty$ or $-\infty$ as $\kappa\to\infty$, while in case~(ii) the normalised tangent of $\mathcal C$ converges as $\kappa\to\infty$.
\end{proposition}
\begin{proof}
\emph{Step 1: the phase portrait.} The rest points, the eigenvalue formula, the exponential asymptotic stability of $\mb v^*$ for $\gamma<\gamma_\sharp$, and the convergence of the gaseous orbit $\mb v(\tau)\to\mb v^*$ with an exponential rate $|\mb v(\tau)-\mb v^*|\lesssim e^{-c\tau}$, $c>0$, are all established in~\cite{Lam_2024} (proven there by a Poincar\'e--Bendixson and Bendixson--Dulac argument). The eigenvalues $\lambda_\pm$ form a non-real conjugate pair exactly when $\mathcal D<0$, giving the focus/node dichotomy. For the claim about $d<10$: as a function of $\gamma$, $\mathcal D(\gamma,d)=(d-2)^2(2-\gamma)^2-8(2(d-1)-d\gamma)$ is a convex quadratic, so its maximum over $[1,\gamma_\sharp]$ is attained at an endpoint; at the endpoints,
\begin{align*}
\mathcal D(1,d)&=(d-2)^2-8(d-2)=(d-2)(d-10),
\\
\mathcal D(\gamma_\sharp,d)&=(d-2)^2\frac{16}{(d+2)^2}-16\frac{d-2}{d+2}=-\frac{64(d-2)}{(d+2)^2}<0,
\end{align*}
using $2-\gamma_\sharp=4/(d+2)$ and $2(d-1)-d\gamma_\sharp=2(d-2)/(d+2)$. Hence $\mathcal D<0$ on all of $[1,\gamma_\sharp]$ whenever $d<10$.

\emph{Step 2: winding at a focus.} Suppose $\mathcal D<0$, and write $\mb F$ for the vector field of~\eqref{E:DS} (so that $\frac{\d}{\d\tau}\mb v=\mb F(\mb v)$), $A:=\grad\mb F(\mb v^*)$, $\alpha:=\Re\lambda_\pm=2/(2-\gamma)-1-d/2<0$ and $\beta:=\Im\lambda_+=\frac12\sqrt{-\mathcal D}\,(2-\gamma)^{-1}\neq0$. By the real Jordan form there is an invertible matrix $P$ with $PAP^{-1}=\alpha I+\beta J$, where $J$ is the matrix of rotation by $\pi/2$. Let $\mb x(\tau):=P(\mb v(\tau)-\mb v^*)$, then
\begin{align*}
\frac{\d\mb x}{\d\tau}=P\mb F(\mb v)=(\alpha I+\beta J)\mb x+\mb g(\mb x),
\qquad\text{where}\qquad
\abs{\mb g(\mb x)}\le C|\mb x|^2 \ \text{ near } \mb 0,
\end{align*}
since $\mb F$ is real-analytic (in particular $C^2$) near $\mb v^*$ (where $v_1>0$) and so $\mb F(\mb v)=A(\mb v-\mb v^*)+O(|\mb v-\mb v^*|^2)$. The gaseous orbit is not the rest point itself ($v_1(\tau)\to0$ as $\tau\to-\infty$), so by uniqueness for the $C^1$ field it never reaches it: $r(\tau):=|\mb x(\tau)|>0$ for all $\tau$, while $r(\tau)\lesssim e^{-c\tau}\to0$ by Step~1. In polar coordinates $\mb x=r(\cos\theta,\sin\theta)$, and writing $\mb x\wedge\mb y:=x_1y_2-x_2y_1$, we have $\mb x\wedge(\alpha I+\beta J)\mb x=\beta r^2$, whence
\begin{align*}
\frac{\d\theta}{\d\tau}=\frac{\mb x\wedge\frac{\d}{\d\tau}\mb x}{r^2}=\beta+\frac{\mb x\wedge\mb g(\mb x)}{r^2}=\beta+O(r(\tau)).
\end{align*}
Since $\int^\infty r(\tau)\d\tau<\infty$, integrating gives $\theta(\tau)=\beta\tau+c_0+o(1)$: the polar angle of $\mb x$ tends to $\pm\infty$ linearly. For the velocity, identify $\R^2\cong\C$ so that $\alpha I+\beta J$ is multiplication by $\alpha+i\beta$; then $\arg((\alpha I+\beta J)\mb x)=\theta+\arg(\alpha+i\beta)$ exactly, while $|\mb g(\mb x)|/|(\alpha I+\beta J)\mb x|\le Cr/\sqrt{\alpha^2+\beta^2}\to0$, so a continuous polar angle of $\frac{\d}{\d\tau}\mb x$ equals $\theta(\tau)+\arg(\alpha+i\beta)+o(1)$ and also tends to $\pm\infty$ linearly.

It remains to remove the linear map $P$. If $Q$ is any fixed invertible $2\times2$ matrix, the induced map $\mb u\mapsto Q\mb u/\abs{Q\mb u}$ is a homeomorphism of the unit circle, of degree $\pm1$; any lift $\tilde q:\R\to\R$ of it therefore satisfies $\tilde q(\theta+2\pi)=\tilde q(\theta)\pm2\pi$, so $\tilde q(\theta)\mp\theta$ is bounded and $\tilde q(\theta)\to\pm\infty$ as $\theta\to+\infty$ (and correspondingly as $\theta\to-\infty$). Applying this with $Q=P^{-1}$ to the unit vectors $\mb x/|\mb x|$ and $\frac{\d}{\d\tau}\mb x/|\frac{\d}{\d\tau}\mb x|$ shows that the polar angles of $\mb v(\tau)-\mb v^*=P^{-1}\mb x$ and of $\frac{\d}{\d\tau}\mb v=P^{-1}\frac{\d}{\d\tau}\mb x$ tend to $+\infty$ or $-\infty$ linearly as well. This proves~(i).

\emph{Step 3: asymptotic direction at a node.} Suppose $\mathcal D\ge0$, so that $\lambda_-\le\lambda_+<0$ are real (negative by the exponential stability of Step~1) and $\mb v^*$ is a stable node of the real-analytic field $\mb F$. By the classical theory of planar hyperbolic nodes~\cite{Hartman_2002}, every orbit converging to $\mb v^*$ does so with a definite limiting secant direction: $(\mb v(\tau)-\mb v^*)/|\mb v(\tau)-\mb v^*|\to\mb e$ for some unit eigenvector $\mb e$ of $A$, say $A\mb e=\lambda\mb e$ with $\lambda\in\{\lambda_-,\lambda_+\}$. Then, since $\frac{\d}{\d\tau}\mb v=A(\mb v-\mb v^*)+O(|\mb v-\mb v^*|^2)$ with $|A(\mb v-\mb v^*)|\ge|\mb v-\mb v^*|/\|A^{-1}\|$,
\begin{align*}
\frac{\frac{\d}{\d\tau}\mb v}{|\frac{\d}{\d\tau}\mb v|}
=\frac{A\frac{\mb v-\mb v^*}{|\mb v-\mb v^*|}+O(\abs{\mb v-\mb v^*})}{\big|A\frac{\mb v-\mb v^*}{|\mb v-\mb v^*|}+O(\abs{\mb v-\mb v^*})\big|}
\to\frac{A\mb e}{|A\mb e|}=\frac{\lambda}{|\lambda|}\mb e=-\mb e ,
\end{align*}
a unit eigenvector of the linearisation. This proves~(ii).

\emph{Step 4: transcription to $\mathcal C$.} By Lemma~\ref{L:orbit}, $(R_\kappa,M_\kappa)=\Phi(\mb v(\tau_\kappa))$ with $\tau_\kappa\nearrow\sup I=\infty$ as $\kappa\to\infty$ (the support is infinite since $\gamma<\gamma_\sharp$). Continuity of $\Phi$ and Step~1 give $(R_\kappa,M_\kappa)\to\Phi(\mb v^*)=(R_\infty,M_\infty)$, which is~\eqref{E:LIMITPOINT}; that $(R_\infty,M_\infty)$ is the radius and mass of the truncated singular profile $\bar\rho_s(r)=v_1^*\,r^{-2/(2-\gamma)}$ is a direct computation: $\bar\rho_s(R)=1$ gives $R=(v_1^*)^{(2-\gamma)/2}=R_\infty$, and
\begin{align*}
4\pi\int_0^{R_\infty}r^{d-1}v_1^*r^{-\frac{2}{2-\gamma}}\d r
=\frac{4\pi v_1^*}{d-\frac{2}{2-\gamma}}R_\infty^{\,d-\frac{2}{2-\gamma}}
=v_2^*(v_1^*)^{\frac{2-\gamma}{2}\brac{d-\frac{2}{2-\gamma}}}
=(v_1^*)^{a}v_2^*=M_\infty,
\end{align*}
using $v_2^*=4\pi v_1^*/(d-\frac{2}{2-\gamma})$ from~\eqref{E:RESTPOINT} and the rest point equation, and $\frac{2-\gamma}{2}(d-\frac{2}{2-\gamma})=\frac{d(2-\gamma)-2}{2}=a$.

For the tangent, differentiating~\eqref{E:CURVEORBIT} in $\kappa$,
\begin{align*}
\brac{R'(\kappa),M'(\kappa)}=\frac{\d\tau_\kappa}{\d\kappa}D\Phi\brac{\mb v(\tau_\kappa)}\frac{\d}{\d\tau}\mb v(\tau_\kappa),
\qquad
\frac{\d\tau_\kappa}{\d\kappa}>0 ,
\end{align*}
so the tangent direction of $\mathcal C$ at $\kappa$ is that of $D\Phi(\mb v(\tau))\mb u(\tau)$ at $\tau=\tau_\kappa$, where $\mb u(\tau):=\frac{\d}{\d\tau}\mb v/|\frac{\d}{\d\tau}\mb v|$. As $\tau\to\infty$, $D\Phi(\mb v(\tau))\to B:=D\Phi(\mb v^*)$, which is invertible; hence $|D\Phi(\mb v(\tau))\mb u(\tau)-B\mb u(\tau)|\to0$ uniformly while $|B\mb u(\tau)|$ is bounded below, so the unit directions of $D\Phi(\mb v(\tau))\mb u(\tau)$ and of $B\mb u(\tau)$ are uniformly close for large $\tau$, and continuous lifts of their polar angles differ by $o(1)$ up to a constant in $2\pi\Z$. In case~(ii), $\mb u(\tau)\to-\mb e$, so the normalised tangent of $\mathcal C$ converges to $-B\mb e/|B\mb e|$. In case~(i), the lifted polar angle of $\mb u(\tau)$ tends to $\pm\infty$ by Step~2, hence, by the circle-homeomorphism argument of Step~2 applied to $Q=B$, so does that of $B\mb u(\tau)$, and therefore so does any continuous polar angle of the tangent of $\mathcal C$; the identical argument applied to the positions, using $\Phi(\mb v(\tau))-\Phi(\mb v^*)=B(\mb v(\tau)-\mb v^*)+O(|\mb v(\tau)-\mb v^*|^2)$ and Step~2 for the polar angle of $\mb v(\tau)-\mb v^*$, shows that the polar angle of $(R_\kappa,M_\kappa)-(R_\infty,M_\infty)$ tends to $\pm\infty$: the curve $\mathcal C$ spirals into $(R_\infty,M_\infty)$.
\end{proof}

The penultimate ingredient is the behaviour of the curve at its two ends: near the incompressible limit $\kappa\searrow1$, where the curve leaves the origin, and, in the compact-support regime, in the large central density limit, where it returns to it.

\begin{lemma}[Endpoints of the mass--radius curve]\label{L:endpoints}
Let $\gamma\in[1,2)$.
\begin{enumerate}[(i)]
\item As $\kappa\searrow1$ one has $S_\kappa^2=(\gamma d/2\pi)\brac{1-1/\kappa}\brac{1+o(1)}$ and $(R_\kappa,M_\kappa)\to(0,0)$, and there is $\kappa_s>1$ such that $R'>0$ and $M'>0$ on $(1,\kappa_s)$.
\item If $\gamma=\gamma_\sharp$, then explicitly
\begin{align}\label{E:SHARPEXPL}
R_\kappa=\frac{d}{\sqrt{2\pi}}\,\kappa^{-\frac{2}{d+2}}\brac{\kappa^{\frac{2}{d+2}}-1}^{\frac12},
\qquad
M_\kappa=\frac{4\pi}{d}\brac{\frac{d^2}{2\pi}}^{\frac d2}\kappa^{-\frac{2(d-1)}{d+2}}\brac{\kappa^{\frac{2}{d+2}}-1}^{\frac d2};
\end{align}
in particular $(R_\kappa,M_\kappa)\to(0,0)$ as $\kappa\to\infty$, and $\kappa\mapsto M_\kappa$ has exactly one critical point, a strict maximum, at $\kappa_1=(2(d-1)/(d-2))^{(d+2)/2}$, with $R'(\kappa_1)<0$.
\item If $\gamma_\sharp<\gamma<\gamma_*$, then $S_\kappa\to S_*<\infty$ and $(R_\kappa,M_\kappa)\to(0,0)$ as $\kappa\to\infty$, and there is $\kappa_\ell>1$ such that $M'<0$ and $R'<0$ for $\kappa\ge\kappa_\ell$.
\end{enumerate}
\end{lemma}
\begin{proof}
Throughout we use $R_\kappa=\kappa^{-\frac{2-\gamma}{2}}S_\kappa$ and $M_\kappa=\kappa^{-a}m_*(S_\kappa)$ from~\eqref{E:RMK}, with $a$ as in~\eqref{E:PHI}, together with $m_*'(S)=4\pi S^{d-1}\bar\rho_*(S)$ and, from differentiating $\bar\rho_*(S_\kappa)=1/\kappa$,
\begin{align}\label{E:SKPRIME}
\frac{\d S_\kappa}{\d\kappa}=\frac{1}{\kappa^2|\partial_y\bar\rho_*(S_\kappa)|}>0.
\end{align}

(i) Integrating the hydrostatic identity $\partial_y\bar\rho_*^\gamma=-4\pi\bar\rho_*y^{1-d}\int_0^ys^{d-1}\bar\rho_*\d s=-(4\pi/d)y(1+O(y^2))$ near the origin gives $\bar\rho_*^\gamma(y)=1-(2\pi/d)y^2+O(y^4)$, hence
\begin{align*}
\bar\rho_*(y)=1-\frac{2\pi}{\gamma d}y^2+O(y^4),
\qquad
\partial_y\bar\rho_*(y)=-\frac{4\pi}{\gamma d}y(1+O(y^2)).
\end{align*}
Solving $\bar\rho_*(S_\kappa)=1/\kappa$ for $S_\kappa\to0$ yields the stated asymptotics of $S_\kappa^2$, and then $R_\kappa\to0$, $m_*(S_\kappa)=(4\pi/d)S_\kappa^d(1+o(1))\to0$, so $M_\kappa\to0$. By~\eqref{E:SKPRIME}, $\frac{\d}{\d\kappa}S_\kappa=(\gamma d/4\pi)\kappa^{-2}S_\kappa^{-1}(1+o(1))\to+\infty$, so
\begin{align*}
R'=\kappa^{-\frac{2-\gamma}{2}}\brac{\frac{\d S_\kappa}{\d\kappa}-\frac{2-\gamma}{2\kappa}S_\kappa}>0
\end{align*}
for $\kappa$ close to $1$, while
\begin{align*}
M'=\kappa^{-a}\brac{m_*'(S_\kappa)\frac{\d S_\kappa}{\d\kappa}-\frac a\kappa m_*(S_\kappa)}
=\kappa^{-a}\brac{\gamma d\kappa^{-3}S_\kappa^{d-2}\brac{1+o(1)}+O(S_\kappa^d)}>0
\end{align*}
for $\kappa$ close to $1$, since $S_\kappa^d=o(S_\kappa^{d-2})$; here we used $m_*'(S_\kappa)=4\pi S_\kappa^{d-1}\bar\rho_*(S_\kappa)=4\pi S_\kappa^{d-1}\kappa^{-1}$.

(ii) At $\gamma=\gamma_\sharp$ the unit gaseous profile is explicit~\cite{Lam_2024}: $\bar\rho_*(s)=(1+(2\pi/d^2)s^2)^{-(d+2)/2}$. From the antiderivative identity
\begin{align*}
\frac{\d}{\d s}\sbrac{s^d(1+cs^2)^{-\frac d2}}=ds^{d-1}(1+cs^2)^{-\frac{d+2}{2}},
\qquad c:=\frac{2\pi}{d^2},
\end{align*}
we get $m_*(S)=(4\pi/d)S^d(1+cS^2)^{-d/2}$. The relation $\bar\rho_*(S_\kappa)=1/\kappa$ reads $1+cS_\kappa^2=\kappa^{2/(d+2)}=:x$, i.e. $S_\kappa^2=(d^2/2\pi)(x-1)$, and substituting into~\eqref{E:RMK}, with $(2-\gamma_\sharp)/2=2/(d+2)$ and $a=(d-2)/(d+2)$, yields~\eqref{E:SHARPEXPL}. In the variable $x=\kappa^{2/(d+2)}>1$, which is an increasing function of $\kappa$,
\begin{align*}
M=\text{const}\cdot(x-1)^{\frac d2}\,x^{-(d-1)},
\qquad
\frac{\d M}{\d x}=\text{const}\cdot(x-1)^{\frac d2-1}x^{-d}\brac{(d-1)-\frac{d-2}{2}x},
\end{align*}
so $M$ is strictly increasing for $x<2(d-1)/(d-2)$ and strictly decreasing for $x>2(d-1)/(d-2)$: the unique critical point is the strict maximum at $x_1=2(d-1)/(d-2)$, i.e.\ at $\kappa_1=\smash{x_1^{(d+2)/2}}$. Similarly $R^2=(d^2/2\pi)(x-1)x^{-2}$ satisfies $\frac{\d}{\d x}[(x-1)x^{-2}]=x^{-3}(2-x)$, so $R$ is strictly decreasing for $x>2$; since $x_1=2+2/(d-2)>2$, we get $R'(\kappa_1)<0$. Finally $M\sim x^{1-d/2}\to0$ and $R^2\sim(d^2/2\pi)x^{-1}\to0$ as $x\to\infty$.

(iii) Here $\gamma>\gamma_\sharp>1$ and $\bar\rho_*$ has compact support $[0,S_*]$. The enthalpy $w:=\bar\rho_*^{\gamma-1}$ satisfies, by the hydrostatic identity,
\begin{align*}
\partial_yw=\frac{\gamma-1}{\gamma}\,\bar\rho_*^{-1}\partial_y\bar\rho_*^{\gamma}=-\frac{\gamma-1}{\gamma}\,y^{1-d}m_*(y),
\end{align*}
whose right side is continuous up to $S_*$; hence $w\in C^1$ up to the support boundary with
\begin{align*}
\partial_yw(S_*)=-\frac{\gamma-1}{\gamma}\,S_*^{1-d}m_*(S_*)=:-\omega<0.
\end{align*}
The relation $\bar\rho_*(S_\kappa)=1/\kappa$ reads $w(S_\kappa)=\kappa^{-(\gamma-1)}$, so $S_\kappa\to S_*$ as $\kappa\to\infty$ and, differentiating, $\frac{\d}{\d\kappa}S_\kappa=(\gamma-1)\omega^{-1}\kappa^{-\gamma}(1+o(1))$. Then $m_*(S_\kappa)\to m_*(S_*)>0$ and
\begin{align*}
M'&=\kappa^{-a-1}\brac{\kappa m_*'(S_\kappa)\frac{\d S_\kappa}{\d\kappa}-am_*(S_\kappa)}\\
&=\kappa^{-a-1}\brac{4\pi S_*^{d-1}\,\frac{\gamma-1}{\omega}\kappa^{-\gamma}\brac{1+o(1)}-am_*(S_*)\brac{1+o(1)}}<0
\end{align*}
for all large $\kappa$, since $\gamma>1$ and $a>0$ for $\gamma<\gamma_*$; here again $m_*'(S_\kappa)=4\pi S_\kappa^{d-1}\kappa^{-1}$. Likewise
\begin{align*}
R'=\kappa^{-\frac{2-\gamma}{2}-1}\brac{\kappa\frac{\d S_\kappa}{\d\kappa}-\frac{2-\gamma}{2}S_\kappa}
=\kappa^{-\frac{2-\gamma}{2}-1}\brac{O(\kappa^{1-\gamma})-\frac{2-\gamma}{2}S_*(1+o(1))}<0
\end{align*}
for all large $\kappa$. Finally $R_\kappa=\kappa^{-\frac{2-\gamma}{2}}S_\kappa\to0$ and $M_\kappa=\kappa^{-a}m_*(S_\kappa)\to0$.
\end{proof}

The last ingredient is stability at small relative central density; we include the short proof to keep the counting theory self-contained (cf.~\cite{Lam_2024}, where this was first proven).

\begin{lemma}[Stability at small relative central density]\label{L:smallkappa}
There exists $\kappa_0>1$ such that $\nmodes(\kappa)=0$ for all $\kappa\in(1,\kappa_0)$.
\end{lemma}
\begin{proof}
We use the rescaled pencil of Step~1 of the proof of Lemma~\ref{L:jump}. By~\eqref{E:WDEF} and the bounds $1\le\sigma_\kappa\le\kappa$ of~\eqref{E:RHOBOUNDS},
\begin{align*}
\sup_{[0,1]}\,\abs{W_\kappa}\le 4\pi|\gamma_*-\gamma|\kappa^2R_\kappa^2\to 0
\qquad\text{as}\quad\kappa\searrow1,
\end{align*}
since $R_\kappa\to0$ by Lemma~\ref{L:endpoints}(i). Next, for $\tilde\chi\in X=H^1([0,1],z^{d+1})$ we claim the weighted Poincar\'e inequality
\begin{align}\label{E:POINCARE}
\int_0^1z^{d+1}\tilde\chi^2\d z\le\frac{2}{d+2}\tilde\chi(1)^2+\frac1d\int_0^1z^{d+1}(\partial_z\tilde\chi)^2\d z.
\end{align}
Indeed, writing $\tilde\chi(z)=\tilde\chi(1)-\int_z^1\partial_z\tilde\chi\,\d s$ and using $(a+b)^2\le2a^2+2b^2$ followed by the Cauchy--Schwarz inequality
\begin{align*}
\brac{\int_z^1\abs{\partial_z\tilde\chi}\d s}^2
\le\int_z^1s^{d+1}(\partial_z\tilde\chi)^2\d s\int_z^1s^{-(d+1)}\d s
\le\frac{z^{-d}}{d}\int_0^1s^{d+1}(\partial_z\tilde\chi)^2\d s,
\end{align*}
we get, upon multiplying by $z^{d+1}$ and integrating (with $\int_0^1z^{d+1}\d z=\frac1{d+2}$ and $\int_0^1z\,\d z=\frac12$), the inequality~\eqref{E:POINCARE}. Since $\sigma_\kappa\ge1$ and the boundary term of~\eqref{E:QLDEF} is non-negative, \eqref{E:POINCARE} gives, for every $\tilde\chi\in X$,
\begin{align*}
Q_\kappa[\tilde\chi,\tilde\chi]
&\ge\gamma\int_0^1z^{d+1}(\partial_z\tilde\chi)^2\d z+d\gamma\tilde\chi(1)^2-\sup_{[0,1]}\abs{W_\kappa}\int_0^1z^{d+1}\tilde\chi^2\d z\\
&\ge\brac{\gamma-\frac{\sup\abs{W_\kappa}}{d}}\int_0^1z^{d+1}(\partial_z\tilde\chi)^2\d z+\brac{d\gamma-\frac{2\sup\abs{W_\kappa}}{d+2}}\tilde\chi(1)^2.
\end{align*}
Choose $\kappa_0>1$ so small that $\sup_{[0,1]}\abs{W_\kappa}<d\gamma$ for $\kappa\in(1,\kappa_0)$; then both coefficients are positive, and $Q_\kappa[\tilde\chi,\tilde\chi]>0$ for every $\tilde\chi\in X\setminus\set0$ (if both terms vanish, $\tilde\chi$ is a constant vanishing at $z=1$, hence zero). By~\eqref{E:PULLBACK}, $\<\L_\kappa\chi,\chi\>=R_\kappa^d\,Q_\kappa[\tilde\chi,\tilde\chi]>0$ for every $\chi\in H_\kappa\setminus\set0$, so $\nmodes(\kappa)=\Mor(\L_\kappa)=0$ for $\kappa\in(1,\kappa_0)$.
\end{proof}

With all the lemmas in place we prove the turning point principle.

\begin{proof}[Proof of Theorem~\ref{T:TPP}]
Throughout, $\nmodes(\kappa)=\Mor(\L_\kappa)$ by Theorem~\ref{T:COUNT}, mass critical points are isolated, $\brac{R'(\kappa),M'(\kappa)}\neq(0,0)$ for every $\kappa$, and~\eqref{E:MORLOC} holds near every mass critical point, all by Lemma~\ref{L:jump}. It is convenient to prove parts~(i), (ii)(a) and~(iii) first, and part~(ii)(b) last.

\emph{Part (i).} On an interval free of mass critical points, $\Mor(\L_\kappa)$ is constant by Lemma~\ref{L:jump}. If $\kappa_0$ is not a mass critical point then, $M'$ being continuous and the critical points isolated, a whole neighbourhood of $\kappa_0$ is free of mass critical points and the count is constant there; hence the count can change only at mass critical points. At a mass critical point with $\partial_\kappa^2M_\kappa\neq0$, Lemma~\ref{L:jump} gives the jump $\pm1$ according to the direction of the sign change of the bending product, equal to the jump of the turning index $i_\kappa$. The geometric formulation is the final sentence of Lemma~\ref{L:winding}: a growing mode is gained exactly when the tangent of $\mathcal C$ crosses the horizontal counter-clockwise, i.e.\ when $\mathcal C$ bends counter-clockwise at the mass extremum, and lost at a clockwise crossing.

\emph{Part (ii)(a).} Let $\gamma\ge\gamma_*$. Then $\gamma_*-\gamma\le0$ while $\partial_y\bar\rho_\kappa^\gamma<0$ by Lemma~\ref{L:PROFBND}, so all three terms of the energy form~\eqref{E:TE} are non-negative: $\<\L_\kappa\chi,\chi\>\ge0$ for every $\chi\in H_\kappa$ and every $\kappa$, whence $\Mor(\L_\kappa)=0$ and $\nmodes\equiv0$ (cf.~\cite{Lam_2024}). For the monotonicity, in~\eqref{E:RMK} the prefactor $\kappa^{\frac12(d\gamma-2(d-1))}$ is positive and non-decreasing (the exponent is $\ge0$), while $\kappa\mapsto m_*(S_\kappa)$ is positive and strictly increasing, since $S_\kappa=\bar\rho_*^{-1}(1/\kappa)$ is strictly increasing and $m_*'=4\pi S^{d-1}\bar\rho_*>0$ on the interior of the support. Hence $\kappa\mapsto M_\kappa$ is strictly increasing and $\mathcal C$ has no turning point.

\emph{Part (iii)(a).} Let $\gamma<\gamma_\sharp$ and $\mathcal D(\gamma,d)<0$; by Proposition~\ref{P:geometry} this holds automatically when $d<10$. Let $\varphi:(1,\infty)\to\R$ be the continuous tangent angle of Lemma~\ref{L:winding}. By Proposition~\ref{P:geometry}(i) and its final assertion, $\varphi(\kappa)\to+\infty$ or $\varphi(\kappa)\to-\infty$ as $\kappa\to\infty$. Fix a non-critical $a_0\in(1,\kappa_s)$ with $\kappa_s$ as in Lemma~\ref{L:endpoints}(i). For every non-critical $b$, Lemma~\ref{L:winding} gives
\begin{align*}
\nmodes(b)=\nmodes(a_0)+\left\lfloor\frac{\varphi(b)}{\pi}\right\rfloor-\left\lfloor\frac{\varphi(a_0)}{\pi}\right\rfloor.
\end{align*}
Were $\varphi(\kappa)\to-\infty$, the right side would tend to $-\infty$ along non-critical $b\to\infty$, contradicting $\nmodes\ge0$. Hence $\varphi(\kappa)\to+\infty$, and $\nmodes(b)\to\infty$ along non-critical $b\to\infty$. At a mass critical point $\kappa_0$, \eqref{E:MORLOC} gives $\nmodes(\kappa_0)=\Mor(\L_{\kappa_0})\ge\nmodes(b)-1$ for nearby non-critical $b$, so in fact $\nmodes(\kappa)\to\infty$ unrestrictedly, which is~\eqref{E:INFTY}. Since $\varphi$ is continuous and tends to $+\infty$, it crosses every sufficiently large level $k\pi$, $k\in\Z$, so $M'=\abs{(R',M')}\sin\varphi$ vanishes at infinitely many parameters: $\mathcal C$ has infinitely many turning points. Finally, $\mathcal C$ spirals into $(R_\infty,M_\infty)$ by Proposition~\ref{P:geometry}.

\emph{Part (iii)(b).} Let $\gamma<\gamma_\sharp$ and $\mathcal D(\gamma,d)\ge0$; by Proposition~\ref{P:geometry} this forces $d\ge10$. By Lemma~\ref{L:orbit}, $M_\kappa=F(\mb v(\tau_\kappa))$ with the mass observable $F(v_1,v_2):=v_1^{a}v_2$, so
\begin{align}\label{E:MASSALONG}
M'(\kappa)=\frac{\d\tau_\kappa}{\d\kappa}\grad F(\mb v(\tau))\cdot\frac{\d}{\d\tau}\mb v(\tau)\Big|_{\tau=\tau_\kappa}
=\frac{\d\tau_\kappa}{\d\kappa}\,\Big|\frac{\d}{\d\tau}\mb v(\tau_\kappa)\Big|\grad F(\mb v(\tau_\kappa))\cdot\mb u(\tau_\kappa),
\end{align}
where $\mb u:=\frac{\d}{\d\tau}\mb v/|\frac{\d}{\d\tau}\mb v|$ and $\frac{\d}{\d\kappa}\tau_\kappa>0$. By Proposition~\ref{P:geometry}(ii), $\mb u(\tau)\to\mb e_0$, a unit eigenvector of $A=\grad\mb F(\mb v^*)$, while $\grad F(\mb v(\tau))\to\grad F(\mb v^*)$. We claim
\begin{align}\label{E:NOTANGENCY}
\grad F(\mb v^*)\cdot\mb e\neq0\qquad\text{for \emph{every} eigenvector $\mb e$ of $A$.}
\end{align}
Granting~\eqref{E:NOTANGENCY}, the right side of~\eqref{E:MASSALONG} has a constant sign for all large $\kappa$, so the mass critical points do not accumulate at $\infty$; being isolated (Lemma~\ref{L:jump}), they are finitely many, so $\mathcal C$ has finitely many turning points, and by Lemma~\ref{L:jump} the count $\nmodes(\kappa)$ is constant beyond the last of them: it stabilises to a finite value. The convergence of $\mathcal C$ to $(R_\infty,M_\infty)$ with an asymptotic tangent is Proposition~\ref{P:geometry}.

To prove~\eqref{E:NOTANGENCY}, note first that $\grad F(\mb v^*)=(a(v_1^*)^{a-1}v_2^*,(v_1^*)^{a})$ is parallel to $\mb w:=(av_2^*/v_1^*,1)$. From~\eqref{E:RESTPOINT} and the rest point equation $4\pi v_1^*=(d-2/(2-\gamma))v_2^*$ we get $v_2^*/v_1^*=4\pi/b$ with $b:=d-2/(2-\gamma)$, and since $2a=b(2-\gamma)$,
\begin{align*}
w_1=\frac{av_2^*}{v_1^*}=\frac{4\pi a}{b}=2\pi(2-\gamma).
\end{align*}
Next, for a $2\times2$ matrix $A=(A_{ij})$ and a non-zero vector $\mb w$, denote by $J$ the rotation by $\pi/2$, so that the vectors orthogonal to $\mb w$ are the multiples of $J\mb w=(-w_2,w_1)$; then $\mb w$ is orthogonal to an eigenvector of $A$ if and only if $J\mb w$ is an eigenvector, i.e.\ if and only if $\det(J\mb w,AJ\mb w)=0$, and expanding the determinant,
\begin{align*}
\det(J\mb w,AJ\mb w)=-\brac{A_{12}w_1^2+(A_{22}-A_{11})w_1w_2-A_{21}w_2^2}=:-Q(\mb w).
\end{align*}
Here, computing the linearisation of~\eqref{E:DS} at $\mb v^*$ (using $\frac1\gamma(v_1^*)^{1-\gamma}v_2^*=\frac{2}{2-\gamma}$ and $\frac{2-\gamma}{\gamma}(v_1^*)^{2-\gamma}=\frac{b}{2\pi}$, both from~\eqref{E:RESTPOINT}),
\begin{align*}
A=\mat{\displaystyle\frac{2}{2-\gamma}-2 &\displaystyle -\frac{a}{\pi(2-\gamma)^2}\\[1mm] 4\pi & -b},
\end{align*}
and therefore, with $\mb w=\brac{2\pi(2-\gamma),1}$ and using $b(2-\gamma)=2a$ repeatedly,
\begin{align*}
Q(\mb w)&=-\frac{a}{\pi(2-\gamma)^2}4\pi^2(2-\gamma)^2+\brac{-b-\frac{2}{2-\gamma}+2}2\pi(2-\gamma)-4\pi\\
&=-4\pi a-4\pi a-4\pi+4\pi(2-\gamma)-4\pi
=-4\pi(2a+\gamma)
=-4\pi(2(d-1)-d\gamma+\gamma)\\
&=-4\pi(d-1)(2-\gamma)\neq0.
\end{align*}
Hence $\mb w$, and with it $\grad F(\mb v^*)$, is orthogonal to no eigenvector of $A$, which is~\eqref{E:NOTANGENCY}.

\emph{Part (iii)(c).} Let $\gamma_\sharp\le\gamma<\gamma_*$. Suppose first $\gamma=\gamma_\sharp$. By Lemma~\ref{L:endpoints}(ii) the mass has the single critical point $\kappa_1=(2(d-1)/(d-2))^{(d+2)/2}$, a strict maximum with $R'(\kappa_1)<0$; the bending product $M'R'$ therefore passes from $(+)(-)<0$ to $(-)(-)>0$ across $\kappa_1$, so by Lemma~\ref{L:jump} the count jumps by $+1$ there and is constant elsewhere. By Lemma~\ref{L:smallkappa}, $\nmodes=0$ near $\kappa=1$, whence $\nmodes(\kappa)=0$ for $1<\kappa<\kappa_1$, also $\nmodes(\kappa_1)=\Mor(\L_{\kappa_1})=0$ by~\eqref{E:MORLOC} (equal to the value on the side where $M'R'<0$), and $\nmodes(\kappa)=1$ for $\kappa>\kappa_1$: the count stabilises to $1$. The endpoint limits $(R_\kappa,M_\kappa)\to(0,0)$ at both ends are Lemma~\ref{L:endpoints}(i)--(ii).

Suppose now $\gamma_\sharp<\gamma<\gamma_*$. By Lemma~\ref{L:endpoints}(i) and~(iii), $M'>0$ on $(1,\kappa_s)$ and $M'<0$ on $[\kappa_\ell,\infty)$, so all mass critical points lie in the compact interval $[\kappa_s,\kappa_\ell]$; being isolated (Lemma~\ref{L:jump}), they are finitely many, and $\mathcal C$ has finitely many turning points, with $(R_\kappa,M_\kappa)\to(0,0)$ at both ends of the family: $\mathcal C$ is a single arc leaving and returning to the origin. Beyond the last critical point the count is constant (Lemma~\ref{L:jump}); call this value $\nmodes(\infty)$. Since $M'$ is positive near $\kappa=1$ and negative for large $\kappa$, and has a constant sign between consecutive critical points, the number of critical points at which $M'$ changes sign is \emph{odd}. By Lemma~\ref{L:jump} the count jumps by $\pm1$ at each sign-changing critical point (there $R'\neq0$, so the bending product changes sign with $M'$) and by $0$ at every other critical point. Using again Lemma~\ref{L:smallkappa} to start the count at $\nmodes=0$ near $\kappa=1$, the stabilised value $\nmodes(\infty)$ is a sum of an odd number of terms $\pm1$, hence odd; being non-negative, $\nmodes(\infty)\ge1$, so the stars are linearly unstable at all sufficiently large central densities.

\emph{Part (ii)(b).} Let $\gamma<\gamma_*$, and suppose $\gamma\ge\gamma_\sharp$ or $\mathcal D(\gamma,d)<0$; by Proposition~\ref{P:geometry} the latter holds automatically when $d<10$ and $\gamma<\gamma_\sharp$, so the hypothesis contains the case $d<10$. By Lemma~\ref{L:endpoints}(i), $M'>0$ on $(1,\kappa_s)$, and by Lemma~\ref{L:smallkappa}, after shrinking $\kappa_s$, also $\nmodes=0$ on $(1,\kappa_s)$. By part~(iii) --- case~(a) when $\gamma<\gamma_\sharp$ (where then $\mathcal D<0$), and case~(c) when $\gamma_\sharp\le\gamma<\gamma_*$ --- we have $\nmodes(\kappa)\ge1$ for all sufficiently large $\kappa$, recovering --- and, when $d\ge10$ and $\gamma\in[\gamma_\sharp,\gamma_*)$, sharpening --- the large central density instability of~\cite{Lam_2024}. In particular $\nmodes$ is non-constant, so by Lemma~\ref{L:jump} there exists a mass critical point with a non-zero jump, i.e.\ one at which the bending product, equivalently $M'$ (as $R'\neq0$ there), changes sign.

Since $M'>0$ on $(1,\kappa_s)$, all mass critical points lie in $[\kappa_s,\infty)$, and being isolated they form a discrete, well-ordered subset of $[\kappa_s,\infty)$. Let $\kappa_1$ be the least mass critical point at which $M'$ changes sign; by the previous paragraph $\kappa_1$ exists. On $(1,\kappa_1)$ the derivative $M'$ is positive off the finitely many earlier critical points, at each of which it does not change sign; so $\kappa_1$ is the first extremum of the mass and, the sign change there being from $+$ to $-$, a strict local maximum. Every jump of the count on $(1,\kappa_1)$ vanishes (Lemma~\ref{L:jump}), so $\nmodes=0$ on $(1,\kappa_1)$, and moreover $\nmodes(\kappa)=0$ \emph{at} each of these critical points and at $\kappa_1$ itself: indeed by~\eqref{E:MORLOC} the value at an isolated critical point equals the value on the punctured sides where $M'R'<0$, and a side with $M'R'>0$ would carry the value $\Mor(\L_{\kappa_0})+1\ge1$, which is excluded on $(1,\kappa_1]$ where the neighbouring values vanish. This argument also forces $R'(\kappa_1)<0$: were $R'(\kappa_1)>0$, the side $\kappa<\kappa_1$ would have $M'R'>0$ and hence, by~\eqref{E:MORLOC}, $\nmodes=\Mor(\L_{\kappa_1})+1\ge1$ there, contradicting $\nmodes=0$. Thus $\nmodes(\kappa)=0$ for $1<\kappa\le\kappa_1$, and just beyond $\kappa_1$ the bending product is $M'R'=(-)(-)>0$, so by~\eqref{E:MORLOC} $\nmodes=\Mor(\L_{\kappa_1})+1=1$ there; by Lemma~\ref{L:jump} this value persists on the whole interval $(\kappa_1,\kappa_1')$ up to the following mass critical point $\kappa_1'$. This is~\eqref{E:ONSET}, and $\nmodes(\kappa)\ge1$ for all large $\kappa$ was already established. 
\end{proof}

\begin{proof}[Proof of Corollary~\ref{C:radial}]
For $\gamma\ge\gamma_*$ the statement is Theorem~\ref{T:TPP}(ii)(a). For $\gamma<\gamma_*$ with $\gamma\ge\gamma_\sharp$ or $\mathcal D(\gamma,d)<0$ (in particular, by Proposition~\ref{P:geometry}, whenever $d<10$) Theorem~\ref{T:TPP}(ii)(b) gives linear stability, $\nmodes(\kappa)=0$, exactly on $(1,\kappa_1]$ --- which contains all sufficiently small relative central densities --- together with $\nmodes(\kappa)=1$ on an interval immediately beyond $\kappa_1$ and $\nmodes(\kappa)\ge1$ for all sufficiently large $\kappa$; this recovers the radial (in)stability theorem of~\cite{Lam_2024}, and for $d\ge10$, $\gamma\in[\gamma_\sharp,\gamma_*)$ sharpens the large central density instability of~\cite{Lam_2024} beyond the range $d<10$ treated there. The final counting statement is Lemma~\ref{L:jump} summed over the mass critical points in $(a_0,\kappa]$ for a fixed non-critical $a_0$ close to $1$ (where $\nmodes(a_0)=0$), equivalently the winding formula~\eqref{E:WINDING} of Lemma~\ref{L:winding}.
\end{proof}

\begin{proof}[Proof of Corollary~\ref{C:nonradial}]
As shown in~\cite{Lam_nonradial}, under spherical-harmonic decomposition the quadratic form $\<\mb L_\kappa\ph,\ph\>_{\bar\rho}$ of $\mb L_\kappa$ splits into angular blocks: $\<\mb L_\kappa\bs\theta,\bs\theta\>_{\bar\rho}=\sum_{l\ge 0}Q_l[\bs\theta]$. Those of degree $l\ge1$ are non-negative regardless of $\gamma$ and $\kappa$ --- the momentum kernel $\Span\set{\mb e_1,\mb e_2,\mb e_3}$ sitting inside the $l=1$ block as its kernel, hence not affecting its sign --- while the $l=0$ block equals the radial energy form~\eqref{E:TE},
\begin{align}\label{E:ZEROBLOCK}
Q_0[\bs\theta]=\<\L_\kappa\chi_{\bs\theta},\chi_{\bs\theta}\>,
\end{align}
where $\bs\theta\mapsto\chi_{\bs\theta}$ is the linear map of~\cite{Lam_nonradial} carrying
$\bs\theta\in\mathbb H_\kappa$ to a radial amplitude $\chi_{\bs\theta}\in H_\kappa$.

We have $\Mor(\mb L_\kappa)\le\Mor(\L_\kappa)$. To see this, suppose $W\subset\mathbb H_\kappa$ is a subspace on
which $\<\mb L_\kappa\ph,\ph\>_{\bar\rho}$ is negative definite. If $\chi_{\bs\theta}=0$ for some $\bs\theta\in W$, only the blocks of degree $l\ge1$ survive and $\<\mb L_\kappa\bs\theta,\bs\theta\>_{\bar\rho}\ge0$, forcing $\bs\theta=\mb0$; so $\bs\theta\mapsto\chi_{\bs\theta}$ is injective on $W$ and, by~\eqref{E:ZEROBLOCK}, its image is a subspace of $H_\kappa$ of dimension $\dim W$ on which $\<\L_\kappa\ph,\ph\>$ is negative definite.

We have $\Mor(\mb L_\kappa)\ge\Mor(\L_\kappa)$. To see this, put $N:=\Mor(\L_\kappa)$ and let $\chi_1,\dots,\chi_N$ be eigenfunctions of~\eqref{E:GEVP} for the negative eigenvalues $\mu_1\le\dots\le\mu_N<0$ (Lemma~\ref{L:discrete}); by Remark~\ref{R:weakstrong} each is smooth and satisfies the Robin condition~\eqref{E:Robin}, hence so does every $\chi$ in $V:=\Span\set{\chi_1,\dots,\chi_N}$. For such $\chi$ the radial field $\bs\theta_\chi:=(4\pi)^{-1/2}r\chi\mb e_r$ has $(\grad\cdot\bs\theta_\chi)|_{\partial B_R}=0$, so $\bs\theta_\chi\in\mathbb H_\kappa$, and $\chi_{\bs\theta_\chi}=\chi$; thus $\chi\mapsto\bs\theta_\chi$ maps $V$ injectively onto an $N$-dimensional subspace of $\mathbb H_\kappa$ on which, by~\eqref{E:ZEROBLOCK}, $\<\mb L_\kappa\bs\theta_\chi,\bs\theta_\chi\>_{\bar\rho}=\<\L_\kappa\chi,\chi\><0$.

Hence $\Mor(\mb L_\kappa)=\Mor(\L_\kappa)$. A growing mode $e^{\lambda t}\bs\theta$ with $\lambda>0$ solves $\mb L_\kappa\bs\theta=-\lambda^2\bs\theta$, so the span of any set of growing modes carries a negative definite form, which means $\nmodes_{\mathrm{full}}(\kappa)\le\Mor(\mb L_\kappa)$. Also, the lift $\chi\mapsto\bs\theta_\chi$ turns a radial growing mode into a growing mode of the full system, so $\nmodes(\kappa)\le\nmodes_{\mathrm{full}}(\kappa)$. Since $\Mor(\L_\kappa)=\nmodes(\kappa)$ by Theorem~\ref{T:COUNT}, we now have
\[\nmodes(\kappa)\le\nmodes_{\mathrm{full}}(\kappa)\le\Mor(\mb L_\kappa)=\Mor(\L_\kappa)=\nmodes(\kappa),\]
so all four agree; the counts agreeing with multiplicity, every growing mode is radial.
\end{proof}

\section*{Acknowledgements}

The author is supported by NWO grants VI.Vidi.223.019 and OCENW.M20.194, and he thanks Mahir Had\v{z}i\'c for helpful discussions and for introducing him to this problem. As a non-native English speaker, the author used Gemini and Claude to assist with the language and presentation of mathematics, specifically for proofreading and stylistic feedback. All of their outputs were reviewed and confirmed by the author before inclusion.

\bibliographystyle{plain}
\bibliography{tpp-refs}

\end{document}